\pdfoutput=1
\documentclass[12pt]{article}%
\usepackage[utf8]{inputenc}%
\usepackage[tracking=true, letterspace=50, expansion=false]{microtype}

\usepackage{amsmath}%
\usepackage{amsthm}%
\usepackage{amssymb}%

\usepackage{pdflscape}%
\usepackage[margin=1in]{geometry}%
\usepackage{setspace}%
\usepackage[title]{appendix}

\newtheorem{lemma}{Lemma}

\newtheorem{theorem}{Theorem}
\theoremstyle{definition}
\newtheorem{assumption}{Assumption}

\usepackage{mathtools}

\DeclareMathOperator{\var}{var}
\DeclareMathOperator{\diag}{diag}
\DeclareMathOperator{\trace}{tr}
\DeclarePairedDelimiter\abs{\lvert}{\rvert}
\DeclarePairedDelimiter\norm{\lVert}{\rVert}
\DeclarePairedDelimiter\indicatorfence{\{}{\}}
\newcommand\1{\operatorname{I}\indicatorfence}
\renewcommand*{\arraystretch}{1.2}
\newcommand{\yy}{\mathsf{Y}}

\usepackage{xcolor}%
\definecolor{webbrown}{rgb}{.6,0,0}

\usepackage{tabularx}%
\usepackage{siunitx}
\usepackage{booktabs}
\usepackage[normal, online, flushleft]{threeparttable}%
\usepackage[nolist]{acronym}
\begin{acronym}
  \acro{CI}{confidence interval}
  \acro{SB}{sparsity-based}
  \acro{SBE}{sparsity-based estimator}
  \acro{WW2}{World War II}
  \acro{OLS}{ordinary least squares}
\end{acronym}
\usepackage{etoolbox} %
\AtBeginEnvironment{appendices}{%
  \crefalias{section}{appendix}
  \counterwithin{figure}{section}
  \counterwithin{table}{section}
}

\usepackage[
 bibencoding=utf8,%
 maxbibnames=3, %
 minbibnames=1, %
 backend=biber,%
 sortlocale=en_US,%
 style=apa,%
 useprefix=true,%
 doi=true,%
 url=true,
 eprint=true%
]{biblatex}
\DeclareSortingNamekeyTemplate{
  \keypart{\namepart{family}}
  \keypart{\namepart{prefix}}
  \keypart{\namepart{given}}
  \keypart{\namepart{suffix}}
}
\newcommand\LRpntwo[0]{62}
\newcommand\LRnRep[0]{24}
\newcommand\LRnAll[0]{35}
\newcommand\LRnDiff[0]{11}
\newcommand\LRspecRep[0]{287}
\newcommand\LReitherAll[0]{81}

\newcommand\LRdesc[0]{76}

\usepackage{hyperref}%
\hypersetup{%
  breaklinks = true,%
  colorlinks = true,%
  anchorcolor = webbrown,%
  citecolor = webbrown,%
  filecolor = webbrown,%
  linkcolor = webbrown,%
  menucolor = webbrown,%
  urlcolor= webbrown,%
  citebordercolor= 1 0 0,%
  menubordercolor=1 0 0,%
  urlbordercolor=1 0 0,%
  runbordercolor=1 0 0,%
  pdftitle=The Fragility of Sparsity,%
  pdfauthor=Michal Kolesár \& Ulrich Müller \& Sebastian T. Roelsgaard}
\usepackage{cleveref}
\crefname{assumption}{assumption}{assumptions}

\renewcommand{\ss}{s}

\title{The Fragility of Sparsity\thanks{Kolesár acknowledges support by the
    National Science Foundation Grant SES-22049356. All errors are our own. We
    thank Colin Cameron, Helmut Farbmacher, Helmut Küchenhoff, Bentley MacLeod,
    Whitney Newey, Alexandre Poirier, and numerous seminar and conference
    participants for helpful comments, and Max Mongkalakorn for excellent
    research assistance.}} \author{Michal Kolesár \and Ulrich K. Müller \and
  Sebastian T. Roelsgaard}%
\date{\today}

\begin{document}

\maketitle
\pagenumbering{Alph} %
\thispagestyle{empty}
\begin{abstract}
  We show, using three empirical applications, that linear regression estimates
  predicated on the assumption of sparsity are fragile in two ways. First, we
  document that different choices of the regressor matrix which do not impact
  \ac{OLS} estimates, such as the choice of baseline category with categorical
  controls, can move sparsity-based estimates by two standard errors or more.
  Second, we develop two tests of the sparsity assumption by comparing
  sparsity-based estimators with \ac{OLS}. The tests tend to reject the sparsity
  assumption in all three applications. Unless the number of regressors is
  comparable to or exceeds the sample size, \ac{OLS} yields more robust
  inference at little efficiency cost.
\end{abstract}

\clearpage
\pagenumbering{arabic}
\section{Introduction}

The linear regression model is the most common workhorse for estimating causal
effects. The key assumption allowing for a causal interpretation is that,
conditional on the set of controls included in the model, the remaining
variation in the treatment of interest is as good as random. Since the
plausibility of this assumption increases with additional control variables,
researchers commonly estimate regression models with many controls.

Traditional estimation methods like \ac{OLS} limit the dimension $p$ of the
control vector to be smaller than the sample size $n$. This limitation has
fueled the development of alternative estimation methods that allow $p$ to
exceed $n$, provided one is willing to make the assumption of approximate
sparsity. In particular, \textcite{bch14} proposed a post-double-selection
procedure that utilizes lasso regressions to select controls, and several papers
have developed closely related ``debiased lasso'' procedures \parencite[see,
among others][]{JaMo14,vbrd14,ZhZh14}, as well as variants that allow for
treatment effect heterogeneity
\parencite[e.g.,][]{aiw18,bcfh17,ccddhnr18,farrell15}. These developments have
had a large influence on empirical practice: applied researchers now routinely
use such \acp{SBE} even in the regime $p<n$ as a complement to or replacement of
\ac{OLS}. We have surveyed the 50 most highly cited empirical papers that cite
one of these papers, and as \Cref{tab:t_lit1} reports, in only one paper was $p$
larger than $n$. What is more, in \LRpntwo\% of the papers, $p$ was less than
$0.2n$.

We demonstrate that this practice leads to fragile inference, for two
fundamental reasons. First, \acp{SBE}, like the sparsity assumption underlying
them, are not invariant to linear reparametrizations of the control matrix $W$.
Normalization choices---such as the choice of baseline category for a set of
dummy variables, or how and whether to center variables prior to taking powers
or interactions---are immaterial when running \ac{OLS}, but they directly affect
\acp{SBE}. We document that in three empirical applications, these seemingly
innocuous normalization choices drastically alter estimates, inducing variation
of the same magnitude as their standard errors. It is thus impossible to
interpret the resulting inference without substantively engaging with the
question of whether the particular normalization choice used is appropriate. Yet
as reported in column (4) of \Cref{tab:t_lit1}, \LRdesc\% of the surveyed papers
do not tell readers what the form of the control matrix is, let alone argue that
any normalization choices were the right ones to obtain a sparse representation.
While the lack of invariance of \acp{SBE} to linear reparametrizations is
well-known in the theoretical literature, to our knowledge, this paper is the
first to quantify it empirically. Furthermore, our survey indicates that the
lack of invariance is not currently appreciated by researchers in the field.

\begin{table}[tp]
  \renewcommand*{\arraystretch}{1.2}
  \caption{Specification of the control matrix in the
    literature.}\label{tab:t_lit1}
  \centering
  \begin{tabular*}{0.95\linewidth}{@{\extracolsep{\fill}}l S@{}S@{}S@{}S@{} S@{}S@{}S@{}@{}}
    &\multicolumn{1}{c}{$p/n<0.1$} & \multicolumn{1}{c}{$p/n<0.2$}
    & \multicolumn{1}{c}{$p/n<1$} & \multicolumn{1}{c}{Basis specified}\\
    & \multicolumn{1}{c}{(1)} & \multicolumn{1}{c}{(2)} &
     \multicolumn{1}{c}{(3)} & \multicolumn{1}{c}{(4)}\\
    \midrule
    \multicolumn{5}{@{}c}{{A\@: All papers}}\\
    \csname @@input\endcsname ./replication_t1.tex
    \multicolumn{5}{@{}c}{{B\@: Papers with replication code}}\\
    \csname @@input\endcsname ./replication_t2.tex
    \bottomrule
  \end{tabular*}
  \begin{tablenotes}
  \item\footnotesize\emph{Notes}: Panel A reports results for 47 papers that
    used \acp{SBE} out of 50 empirical papers surveyed, while Panel B reports
    results for \LRspecRep\ specifications in a subsample of \LRnRep\ papers
    with available replication code. Cols. 1--3 report the percentage of
    specifications in these papers in which $p/n$ is smaller than $0.1$, $0.2$,
    and $1$, respectively. Col.~4 reports the percentage of specifications in
    which the form of the control matrix $W$ could be ascertained from the
    description in the text. All percentages are weighted, so that each paper
    receives an equal weight. \Cref{sec:deta-liter-revi} details the sample
    construction and variable definitions.
  \end{tablenotes}
\end{table}

Second, one might be skeptical of the plausibility of the sparsity assumption
more generally. In social science applications, it is often difficult to point
to theoretical or institutional arguments for why a small number of controls
should be able to soak up most of the confounding. We thus develop two
statistical tests of the null hypothesis that (approximate) sparsity holds. Both
tests tend to reject in the baseline specification of the three empirical
applications. Furthermore, we are often unable to find any normalization choice
for which the tests do not reject, suggesting that this second concern is not
just a variation of the first.

Consequently, applied researchers should be wary of using \acp{SBE} without a
substantive defense of the sparsity assumption for their specific choice of the
control matrix. Unless $p$ is close to or exceeds $n$, a robust default is to
simply run \ac{OLS}.\footnote{One only needs to be careful in standard error
  construction, as \ac{OLS} no longer consistently estimates the residuals once
  $p\asymp n$. This renders the usual Eicker-Huber-White standard errors
  invalid; instead one needs to use a standard error formula that is robust to
  high-dimensional controls \parencite[see,
  e.g.][]{cjn18,DoSu18,dadamo19,kss20,jochmans22}. If $p$ exceeds $n$, the
  issues we raise are still pertinent; it's just that the appropriate remedy is
  case-specific, as discussed further in \Cref{sec:conclusion}.} While this
recommendation may seem obvious to some, it is practically relevant given that,
as our survey demonstrates, in most empirical specifications employing
\acp{SBE}, $p$ is much smaller than $n$.

Our analysis relates to several strands of literature. Many authors mention the
\acp{SBE}' lack of invariance to linear reparametrization as an undesirable
feature. In the specific context of including categorical controls, alternatives
to standard lasso regression have been developed that are invariant to the
choice of baseline category \parencite[e.g.][]{BoRe09,GeTu10,StShTi21}.
Similarly, one could perhaps incorporate the choice of a centering constant when
taking powers as an additional tuning parameter when fitting \acp{SBE}. However,
as we discuss further in \Cref{sec:discussion} below, such solutions come with
their own challenges. Using a Bayesian framework, \textcite{glp21} present
empirical evidence suggesting sparsity may not be a compelling assumption in
popular data sets used in economics. \Textcite{WuZh23} give simulation evidence
and theoretical arguments showing that even if sparsity arguably holds,
\acp{SBE} may still display substantial bias in finite samples.
\Textcite{AnFr22} explore robustness of \acp{SBE} to tuning parameter choices.

The remainder of the paper is organized as follows. In \Cref{sec:empir-illustr},
we introduce the three empirical examples used to quantify the fragility of
\acp{SBE}: the empirical illustration in \textcite{bch14} on the effect of
abortion on crime, a \textcite{ferrara22} study of occupational upgrading by
Black southerners, and a study of the effect of moral values on voting behavior
by \textcite{enke20}. Our analysis of fragility to linear reparametrizations
focuses on two seemingly innocuous normalizations: first, we consider different
ways of dropping collinear columns, such as which category to drop when
including categorical variables. Second, when $W$ includes powers and
interactions of a set of baseline controls, we consider different ways of
centering the baseline controls, such as demeaning, subtracting the median, or
not centering. We find that these normalizations can move the estimates by two
standard errors or more. We also consider alternative normalizations that
further expand the set of possible specifications of $W$: expressing a set of
categorical controls as indicators for subsets of the categories, rather than
just as indicators for each individual category; and continuously varying where
a variable is centered prior to taking powers, rather than just considering
discrete values such as the mean or the median. These normalizations lead to an
even greater variation in the estimates.

In \Cref{sec:sparse-repr-are}, we interpret these empirical results through the
lens of a thought experiment: suppose a researcher picks one of the possible
normalizations at random. How likely would such a random choice lead to an
(approximately) sparse representation? We compute that probability in three
stylized examples of normalizations and find that it rapidly decreases as a
function of $n$. Leaving default choices for how $W$ is constructed to
statistical software or happenstance is very unlikely to lead to sparsity.

\Cref{sec:effic-gains-under} conducts an analysis of the relative efficiency of
\acp{SBE} and \ac{OLS}. We focus on the regime where $p$ is smaller, but
proportional to $n$. This is because if $p/n\to 0$, \ac{OLS} achieves the
semiparametric efficiency bound under homoskedastic errors, which limits the
arguments in favor of alternative estimators. On the other hand, if $p$ exceeds
$n$, the \ac{OLS} benchmark is no longer available. We show that under
homoskedasticity, the proportional variance reduction of \acp{SBE} relative to
\ac{OLS} is capped at $1-p/n$ (effectively, we can at best avoid a degrees of
freedom adjustment). This benchmark forms the basis for our recommendation to
simply run \ac{OLS} unless $p$ is close to $n$: it yields robust estimates at
little efficiency loss.

In \Cref{sec:testing-sparsity}, we develop two tests of the sparsity assumption,
again in the regime $p\asymp n$.\footnote{\textcite{CaVe21} develop a test in
  the high-dimensional regime $p/n\to\infty$. However, their test requires
  explicit specification of the sparsity level under the null, as does the test
  developed by \textcite{he20} in the regime $p<n$. In the context of factor
  models, \textcite{BeSt24} consider testing the null of a zero vector against a
  sparse alternative.} We circumvent the difficulty that the sparsity assumption
only restricts rates, rather than the actual number of non-zero coefficients for
a given $n$, by testing it indirectly. The first test compares \ac{OLS} and
lasso residuals: under (approximate) sparsity, the residual sum of squares of
the lasso should exceed the \ac{OLS} residual sum of squares only by a small
amount. The second test is a version of the \textcite{hausman78} specification
test: any difference between \ac{SBE} and \ac{OLS} estimates must be explained
by differences in the relative efficiency of the estimators afforded by the
sparsity assumption that \ac{SBE} exploits, but \ac{OLS} does not. Thus, the
common practice of reporting \ac{SBE} estimates alongside \ac{OLS} should not be
interpreted as a robustness check for the \ac{OLS} specification; rather,
divergence between the estimates indicates a failure of the sparsity assumption.

\section{Empirical illustrations}\label{sec:empir-illustr}

In this section, we revisit three applications that leveraged \acp{SBE}: the
application in \textcite[BCH]{bch14} that probes the investigation by
\textcite{DoLe01} of the impact of abortion on crime, the \textcite{ferrara22}
study of employment opportunities for Black southerners in the aftermath of
\ac{WW2}, and the examination of the relationship between moral values and
voting behavior by \textcite{enke20}. In each application, the parameter of
interest is the coefficient $\beta$ in the linear model
\begin{equation}\label{eq:outcome}
  Y_{i}=D_{i}\beta+W_{i}'\gamma+U_{i}, \qquad E[U_{i}\mid D_{i}, W_{i}]=0.
\end{equation}
We estimate $\beta$ by applying the original \ac{SBE} after changing the
specification of the control matrix $W$ (where rows correspond to the control
vectors $W_{i}$) in seemingly innocuous ways that do not impact \ac{OLS}
estimates. We show that the impact of these normalizations on \acp{SBE} is, on
the other hand, substantial. To isolate the effect of the choice of control
matrix from other implementation details, our analysis otherwise sticks to
software defaults.\footnote{For estimates using post-double selection, we use
  the \texttt{hdm} package and its defaults, including the choice of the tuning
  parameter and normalizing the columns of $W$ to have unit variance.}

In \Cref{sec:frag-altern-norm}, we consider the effects of two normalizations.
First, we consider different ways of resolving multicollinearity in the control
matrix by changing which columns are dropped to make the matrix $W$ full rank.
For example, if multicollinearity arises due to the inclusion of categorical
variables, we change which category is dropped. Original implementations of
\acp{SBE} in each of the three applications remove multicollinearity in $W$ as a
data-processing step, similar to standard implementations of least squares
regression. Our exercise is therefore equivalent to changing the order of the
columns of the control matrix, but otherwise conducting the analysis exactly as
in the original. For lasso-based estimators, such a step is typically needed to
ensure that so-called ``compatibility conditions'' or ``restricted eigenvalue''
assumptions hold; these assumptions ensure fast convergence rates for the lasso
\parencite[see, e.g.][]{brt09}.\footnote{A necessary condition for these
  assumptions is that submatrices of $W$ with $2s$ columns, where $s$ is the
  sparsity index, are full rank. It implies that if, say, a categorical variable
  has fewer than $2s$ categories, we cannot include all categories as well as
  the intercept.} Second, when $W$ includes powers and interactions of a
baseline set of controls, we consider different normalizations of the baseline
controls (such as demeaning vs subtracting the median) before taking powers and
interactions. Appropriately centering the baseline variables puts them on the
same scale, and may make the sparsity assumption more plausible and more easily
interpretable. Relative to using raw polynomials, it typically also helps with
numerical stability.

\Cref{sec:other-choic-spec} considers alternative normalizations of $W$ in the
two scenarios above: expressing a set of categorical controls as indicators for
subsets of the categories, rather than only as indicators for each individual
category; and continuously varying where a variable is centered prior to taking
powers and interactions, rather than just considering discrete values such as
the mean or the median.

The specification choices we study here arise commonly in applied papers
employing \acp{SBE}: in our survey, one or both normalization issues arose in
\LReitherAll\% of the 47 papers examined in \Cref{tab:t_lit1}.

\subsection{Normalizations of the control matrix}\label{sec:frag-altern-norm}

We begin by considering different ways of resolving multicollinearity issues
that arise in each application.

The data in the first application consists of an annual panel of US states over
the period 1985--97 originally analyzed by \textcite{DoLe01}, who ran a two-way
fixed effects regression of crime rates on effective abortion rates, state and
year fixed effects, and 8 baseline covariates.\footnote{These are: lags of the
  number of prisoners and police per capita, the unemployment rate, per-capita
  income and beer consumption, poverty rate, AFDC generosity lagged 15 years,
  and a dummy for a shall-issue concealed carry law.} %
BCH argue that this set of 8 controls may be insufficient to purge time-varying
confounders. They consider a first-differences version of this specification,
with 12 time effects and first-differences of the 8 baseline controls, which
they augment with 136 further variables obtained by squaring the baseline
controls and interacting them with each other and with linear and quadratic
trends.\footnote{In particular, BCH add squares of the first differences as well
  as lags and squared lags of the baseline controls, and they interact the first
  differences and their squares with a linear and quadratic trend. They also add
  interactions of the first-differences, and interact them with a linear and a
  quadratic trend. Since the time series of the shall-carry dummy in each state
  never changes from $1$ to $0$, its first difference is also binary. These
  transformations therefore only yield 136 unique columns.
} In addition, they include 49 variables that are time-invariant within each
state, corresponding to initial values and averages of various transformations
of the baseline controls, as well as interactions of these 49 variables with a
time trend and its square. Since there are only 48 states in the data, these
variables span the same column space as state fixed effects. The resulting
control matrix has 303 columns, but rank of only 294: because time effects are
included, 2 of the time-invariant variables are redundant, as are 2 of the
interactions with a time trend and 2 with its square. In addition, one of the
baseline controls, a shall-issue dummy, is binary and non-zero only 21 times,
but it is interacted with 24 variables, so that 3 of the interactions are
redundant. BCH first drop the collinear columns, and then estimate the model by
the post-double lasso estimator that they develop.

\begin{table}
  \renewcommand*{\arraystretch}{1.2}
  \begin{threeparttable}
    \caption{Fragility of sparsity-based methods to normalizations of the control matrix.}\label{tab:t1}
  \begin{tabular}{@{}l@{} S@{}S@{}S@{}S@{}S@{}S@{}}
    &
    & \multicolumn{5}{c}{Sparsity-based estimation}\\
    \cmidrule(rl){3-7}
    & \multicolumn{1}{c}{OLS}
    & \multicolumn{1}{c}{Replication}
    & \multicolumn{2}{c}{Range (collinear)}
    & \multicolumn{2}{c}{Range (powers)}\\
    \cmidrule(rl){4-5}\cmidrule(rl){6-7}
    Outcome & \multicolumn{1}{c}{(1)} & \multicolumn{1}{c}{(2)} & \multicolumn{1}{c}{(3)} & \multicolumn{1}{c}{(4)} & \multicolumn{1}{c}{(5)} & \multicolumn{1}{c}{(6)}\\
    \midrule
      \csname @@input\endcsname ./table1.tex
    \bottomrule
  \end{tabular}
  \begin{tablenotes}
  \item{}\footnotesize\emph{Notes}: Col.~2 replicates \acp{SBE} for the BCH,
    \textcite{ferrara22}, and \textcite{enke20} studies discussed in the text.
    Col.~1 reports the unpenalized \ac{OLS} estimate for the same specification.
    Cols.~3 and 4 report the range of estimates obtained under alternative ways
    of dropping collinear columns of the control matrix. Cols.~5 and 6 report
    the range of estimates obtained under alternative normalizations of the
    controls prior to taking powers and interactions: no normalization,
    demeaning, subtracting the median, and setting the range to $[-1, 1]$ and
    $[0,1]$. The outcome variables in Panel C are multiplied by 100. ``Trump $-$
    avg GOP'' refers to the difference between the indicator for voting for
    Trump in the 2016 general election minus the vote share given to Republican
    candidates in the previous two presidential elections. Standard errors for
    each estimate are given in parentheses (cols.~3--6 report the standard
    errors for the extreme estimates). These are robust for \textcite{enke20},
    clustered by state for BCH and by county for \textcite{ferrara22}.
  \end{tablenotes}
\end{threeparttable}
\end{table}

Column 2 of Panel A in \Cref{tab:t1} replicates the BCH estimates for each of
the three versions of the crime rate outcome variable considered by
BCH.\footnote{The replication differs slightly from the original (Table 2 in
  BCH) because the algorithm for dropping collinear columns in BCH had a coding
  error, yielding a matrix with 296 columns, and a rank of 293.} Comparing the
standard error to the \ac{OLS} standard error in column 1, we see that the
post-double lasso estimator appears to be substantially more precise than the
\ac{OLS} estimate. It is also economically significant: the effective abortion
rate is defined as the average legalized abortion rate among arrestee cohorts
(the number of abortions per live birth in a cohort weighted by the cohort's
share of arrestees, which is outcome-specific). For the violent crime outcome,
its standard deviation across states in 1997 is about 1, so that the estimate
implies a 16\% reduction in crime rate per standard deviation increase in
effective abortion rate. However, this result is very sensitive to how we
resolve the collinearity. To resolve it, we may drop any 3 of the 24
interactions with the shall-issue dummy, and any 2 of the 49 time-invariant
controls; the same holds when we interact these with a time trend and its
square. This gives $\binom{24}{3}\binom{49}{2}^{3}\approx 3\times 10^{12}$
possible ways of obtaining a full rank control matrix. Columns 3 and 4 report a
range of estimates we obtain by randomly choosing among these possibilities,
showing that, depending on the outcome, the estimates move by 1.2 to 1.9
standard errors.

Similar collinearity issues arise in the second application that replicates the
analysis in \textcite{ferrara22}, who studies to what extent post \ac{WW2}
occupational upgrading of Black workers from low-skilled to semi-skilled can be
attributed to war casualties among semi-skilled white soldiers. Using a decennial
1920--1960 unbalanced panel of county-level observations in 16 predominantly
Southern US states, the study runs a two-way fixed effects specification with
county and time fixed effects, interactions between state and time fixed
effects, the share of semi-skilled Black workers as the outcome, and treatment
given by the white casualty rate interacted with a post-war indicator. To purge
time-varying confounders, the study also includes 24 baseline controls
(including the county draft rate, \ac{WW2} spending, demographic and
socioeconomic controls), their squares and interactions, as well as interactions
between the 24 baseline controls and state and time effects. In addition, two
baseline controls (number of slaves in 1860 and the unemployment rate in 1937)
are also included in triple interactions with other controls, time and state
effects. After dropping a reference state and a reference year in each
interaction, and dropping zero and repeated columns, the resulting matrix has
2270 columns, but rank of only 2252, because the state of Delaware only contains
15 observations, but there are 33 Delaware-specific controls.\footnote{In
  particular, since the two baseline controls entering triple interactions are
  time-invariant and the specification includes county effects, we need to
  exclude the main effects of these controls, their squares and
  state-interactions, which yields a control matrix with 2588 columns. 154
  columns are collinear since we need to drop a reference state and year in each
  interaction, and 18 Delaware-specific controls are also collinear. In
  addition, 164 columns are repeated or zero, yielding a rank equal to 2252.}
\textcite{ferrara22} first forms a full-rank control matrix, and then uses
double-$t$ selection to estimate the
model.\footnote{\label{fn:double_t}Specifically, \textcite{ferrara22} first runs
  a selection step that regresses the outcome and the treatment on the control
  matrix using \ac{OLS}, and selects controls with a $t$-statistic that is
  larger than 2.575 in absolute value in either regression. In a second step,
  the outcome is regressed on the treatment and selected controls using
  \ac{OLS}.}

We replicate this specification in column 2 of Panel B in
\Cref{tab:t1}.\footnote{The estimate is slightly higher than in the original
  because, to keep missing data and normalization issues separate, we restrict
  the sample to the 4,903 observations with no missing values in both estimation
  steps described in \cref{fn:double_t}. In addition, we include year fixed
  effects in both the first and second step, rather than exclude them from the
  first step.} However, similar to panel A, there are many ways of resolving
multicollinearity in the control matrix, because there are many ways of
specifying a reference state and a reference year in each interaction, and, if
we keep 3 Delaware county effects, $\binom{30}{18}$ ways of dropping
Delaware-specific controls. As shown in columns 3 and 4, depending on how the
multicollinearity is resolved, the estimates vary by over three standard errors
depending on how we order the state and time effects.

Our final empirical example uses data from \textcite{enke20}, who examines how
voters' moral values affect their voting patterns. Enke uses survey data to
construct an index of the relative importance of universalist moral values
(individual rights, justice, fairness) vs communal values (loyalty, respect). He
then runs a regression of three different measures of voting behavior on this
index, controlling for 10 continuous or binary controls, as well as 5 sets of
categorical variables.\footnote{The continuous or binary controls comprise:
  political liberalism, log of household income, education, log of population
  density of respondent's ZIP code, religiosity, gender, employment indicator,
  altruism, measure of trust, and the absolute value of the moral values index.
  The categorical variables are: county, year of birth, religious denomination
  and occupation fixed effects.} Columns 1 and 2 in panel C of \Cref{tab:t1}
replicate the \ac{OLS} and post-double lasso estimates reported in the paper.
Here the inclusion of multiple sets of categorical variables necessitates a
specification of a reference category for each set. Columns 3 and 4 report the
range of estimates obtained by changing these reference categories, which varies
between a third and a half of the standard error depending on the outcome. This
is a narrower range than in panels A and B, albeit still large enough to affect
the economic interpretation of the estimates.

Columns 5 and 6 of \Cref{tab:t1} show the range of estimates we obtain from
considering different normalizations when taking powers and interactions between
variables in the BCH and \textcite{ferrara22} applications. In addition to no
normalization (the choice in the original analyses), we consider demeaning,
centering at the median, and setting the range to $[-1,1]$ and $[0,1]$. The
table shows that such normalizations can again substantially move the estimates,
by up to 1.3 standard errors.

\subsection{Alternative normalizations of the control
  matrix}\label{sec:other-choic-spec}

We now consider two alternatives to the construction of $W$ in the presence of
categorical variables and power and interactions.

In \Cref{sec:frag-altern-norm}, a categorical variable with $k$ categories was
always specified as a set of $k-1$ dummies for individual categories, with a
reference dummy dropped to prevent collinearity. We now consider expressing such
variables as $k-1$ indicators for different subsets of the $k$ categories. As
discussed in \Cref{sec:categorical-data}, if the subsets are carefully chosen,
such a specification may be more likely to lead to sparsity in certain contexts.
Here we pick the subsets at random, subject to the constraint that they span the
same column space. For the \textcite{enke20} application, columns 3 and 4 of
\Cref{tab:t1} correspond to a special case of this exercise, when the subsets
are constrained to be of cardinality one. Correspondingly, the range of
variation in the estimates over all possible subsets, reported in columns 1 and
2 of \Cref{tab:t2}, is considerably greater, exceeding 0.9 standard errors in
all regressions. For the \textcite{ferrara22} application, the range exceeds two
standard errors.\footnote{The range does not exceed that in columns 3 and 4 of
  \Cref{tab:t1} because we do not vary which Delaware-specific controls are
  dropped.}

\begin{table}
  \renewcommand*{\arraystretch}{1.2}
  \centering
  \begin{threeparttable}
    \caption{Fragility of sparsity-based methods to further normalizations
      of the control matrix.}\label{tab:t2}
  \begin{tabular*}{0.95\linewidth}{@{\extracolsep{\fill}}l@{} S@{}S@{}S@{}S@{}}
    & \multicolumn{4}{c}{Sparsity-based estimation}\\
    \cmidrule(rl){2-5}
    & \multicolumn{2}{c}{Range (category sums)}
    & \multicolumn{2}{c}{Range (offset)}\\
    \cmidrule(rl){2-3}\cmidrule(rl){4-5}
    Outcome & \multicolumn{1}{c}{(1)} & \multicolumn{1}{c}{(2)}
    & \multicolumn{1}{c}{(3)} & \multicolumn{1}{c}{(4)} \\
    \midrule
     \csname @@input\endcsname ./table2.tex
    \bottomrule
  \end{tabular*}
  \begin{tablenotes}
  \item\footnotesize\emph{Notes}: Cols.~1 and 2 report the range of estimates
    obtained under alternative ways of expressing a set of categorical
    variables. Cols.~3 and 4 report the range of estimates obtained when powers
    are constructed using Hermite polynomials, with an offset $\lambda$ ranging
    between $-1$ and $1$. See notes to \Cref{tab:t1} for a description of
    outcome variables. Standard errors for each extreme estimate are given in
    parentheses. These are robust for \textcite{enke20}, clustered by state for
    BCH and by county for \textcite{ferrara22}.
  \end{tablenotes}
\end{threeparttable}
\end{table}

Similarly, the centering of baseline variables at the mean vs the median we
considered in the previous subsection can be thought of as a special case of a
more general problem: we could in principle choose the centering constant
$\lambda$ to be any number. At a theoretical level, we are not aware of
arguments for why $\lambda=0$ (no normalization), or setting it to the sample
mean should more plausibly lead to sparsity than other choices. To further
explore the sensitivity to the choice of $\lambda$, we standardize the baseline
variables, and then vary the specification of $W$ by picking $\lambda$ from the
range $[-1, 1]$ for each baseline variable. To align the exercise with the
analytical calculations in \Cref{sec:set-herm-polyn}, we use Hermite rather than
raw polynomials. Columns 3 and 4 of \Cref{tab:t2} show that the resulting
variation in the estimates is again substantial, exceeding 1.7 standard errors
in all specifications. The sign of the estimated effect is also sensitive to the
choice of $\lambda$ in two of the BCH specifications.

\section{Sparse representations are rare}\label{sec:sparse-repr-are}

In this section we consider the thought experiment of picking a particular linear
representation of the controls $W$ at random. In particular, we consider (i) a full rotation of the column
space; (ii) different ways of controlling for categorical variables; and (iii)
an offset for the scalar variable in a polynomial series regression. We find
that sparse representations are exceedingly rare in each example. This suggests
that in a given application, undiscerning default choices for expressing the
column space will typically not satisfy the sparsity assumptions. This helps to
explain the findings of the previous section: most normalization choices cannot
yield sparse representations. As we discuss in detail in \Cref{sec:discussion}
below, it also implies that the representation of the control matrix is a
substantive choice and not just an implementation detail, like, say, the choice
of the penalty parameter: it impacts not just the finite-sample estimate, but
also the large-sample validity of \ac{SBE}-based inference.

\subsection{Approximate sparsity}\label{sec:approximate-sparsity}

The calculations in this section focus on sparsity in the outcome
regression~\eqref{eq:outcome}. Let $A$ be a full-rank $p\times p$ matrix. Then
an alternative expression for the column space of $W_{i}$ is
$\tilde{W}_{i}=AW_{i}$ with coefficient $\tilde{\gamma}=A'^{-1}\gamma $, so that
$W_{i}^{\prime}\gamma =\tilde{W}%
_{i}^{\prime}\tilde{\gamma}$. Note that any such transformation leaves the
\ac{OLS} coefficient on $D_{i}$ numerically unaltered.

In our examples, $\gamma $ is exactly sparse, with just one non-zero element,
$\norm{\gamma}_{0}=1$. The researcher uses the transformed regressors
$\tilde{W}_{i}$ with coefficients $\tilde{\gamma}$, which typically is not
exactly sparse. However, for sparsity-based inference to be asymptotically
valid, it suffices for $\tilde{\gamma}$ to be \emph{approximately sparse}. In
particular, under our maintained assumption that $p\asymp n$, it suffices that
there exists an $s$-sparse approximation to the control function $W_{i}'\gamma$
in that
\begin{equation}\label{eq:approx_sparsity2}
  \min_{\norm{v}_{0}\leq s} E[(W_{i}^{\prime}\gamma -\tilde{W}_{i}^{\prime}v)^{2}]=O(s/p),
\end{equation}
such that the sparsity index $s$ satisfies
\begin{equation}\label{eq:approx_sparsity1}
s=o(\sqrt{p}/\log p).
\end{equation}
(cf.\ page 614 in \textcite{bch14}; theory underlying other \acp{SBE}, e.g.
\textcite{JaMo14,vbrd14,ZhZh14} imposes analogous or stronger conditions). Here,
and throughout the paper, all limits are taken as $p\rightarrow \infty $ (or,
equivalently, as $n\rightarrow \infty $). The sparsity
condition~\eqref{eq:approx_sparsity1} is stronger than the condition
$s=o(p/\log(p))$, which suffices for consistent estimation of the control
function $W_{i}'\gamma$ \parencite[e.g.,][]{brt09}. However, valid
sparsity-based inference on $\beta$ requires not only that the control function
be estimated consistently, but also that the estimation error in $\gamma$ be
sufficiently small so that the \ac{SBE} of $\beta$ is not asymptotically
biased---the reason we focus on the sparsity condition in
\cref{eq:approx_sparsity1} is that it ensures this is indeed the
case.\footnote{Debiased ``machine learning'' approaches that employ
  cross-fitting \parencite[e.g.,][]{ccddhnr18} allow \cref{eq:approx_sparsity1}
  to be violated provided one imposes a correspondingly stronger condition on
  the sparsity index $s_{p}$ in the propensity score regression,
  $s_{p}=o(\frac{p}{s(\log p)^{2}})$. However, the results in this section imply
  that such sparsity assumptions on the propensity score regression are likewise
  sensitive to normalization choices.}

\subsection{Rotation}

By way of establishing a benchmark, we start with an extreme case, and consider
the set of transformations
\begin{equation*}
  \tilde{W}_{i}=RW_{i}, \quad R\in \mathbf{R}=\{A\in\mathbb{R}^{p\times p}\colon
  A^{\prime}A=I_{p}, \;\det (A)=1\},
\end{equation*}
that is, $\tilde{W}_{i}$ is obtained by rotating $W_{i}$ by $R$, and
$\tilde{\gamma}=R\gamma$. For the purposes of studying sparsity, this is a very
large set of transformations, as for any $\gamma $ with $\norm{\gamma}_{2} >0$,
there exists $R$ such that $\tilde{\gamma}$ has only one non-zero element. At
the same time, there also exists an $R$ such that all elements of
$\tilde{\gamma}$ are equal (and equal to $\norm{\gamma}_{2}/\sqrt{p}$).

Now suppose we take a random draw $\mathcal{R}$ from $\mathbf{R}$, with
distribution equal to the Haar measure (so that for any $R\in \mathbf{R}$,
$R\mathcal{R}\sim \mathcal{R}$). Recall that the multivariate standard normal
distribution is spherically symmetric, that is, for
$\mathcal{Z}\sim \mathcal{N}(0,I_{p})$ and any $R\in \mathbf{R}$,
$R\mathcal{Z}\sim \mathcal{Z}$. This implies that the induced distribution of
$\tilde{\gamma}=\mathcal{R}\gamma $ satisfies
\begin{equation}\label{eq:rot_dist}
  \mathcal{R}\gamma \sim \frac{\norm{\gamma}_{2}}{\norm{\mathcal{Z}}_{2}}\mathcal{Z},
\end{equation}
so the $p$ elements of $\mathcal{R}\gamma$ are, up to a common rescaling,
i.i.d.~standard normal. The normal distribution is thin-tailed, making it very
unlikely that a few largest absolute values of $\mathcal{Z}$ dominate. This
suggests that with very high probability, $\mathcal{R}\gamma$ is not
approximately sparse, as formalized in the following result.

\begin{theorem}\label{thm:rotation}
  Suppose that the eigenvalues of $E[W_{i}W_{i}^{\prime}]$ are bounded away from
  zero and infinity, and that $\norm{\gamma}_{2}\asymp 1$. Then the logarithm of
  the probability that the model with regressor $\tilde{W}_{i}=\mathcal{R}W_{i}$
  satisfies \cref{eq:approx_sparsity1,eq:approx_sparsity2} is of the order
  $-\tfrac{p}{4}\log p$.
\end{theorem}

Given that a sparsity assumption allows for more informative inference, and that
any coefficient vector can be rotated into an (extremely) sparse one, it is not
surprising that a random rotation only rarely yields a sparse representation.
The contribution of \Cref{thm:rotation} is to quantify just how exceedingly rare
this event is: for $p\geq 50$, say, $p^{-p/4}<10^{-21}$. Of course, researchers
rarely consider all rotations when specifying the regressor matrix. The
relevance of this quantification lies in showing that even when the measure of
plausible rotations is orders of magnitude smaller than the full set, the
probability of arriving at a sparse representation is still minuscule.

\subsection{Categorical data}\label{sec:categorical-data}

An empirically common form of controls is dummies for categorical variables. If
there are $p$ underlying categories, and a constant is included, then one must
designate a reference category by dropping one dummy to avoid perfect
collinearity. Suppose the model is exactly sparse with sparsity index $s$ when
one chooses the right reference category. This means that $p-s$ categories have
the same coefficient as the reference category. Thus, if one were to pick a
reference category at random, one has a $1-s/p$ chance of inducing sparsity,
which is a probability close to one under~\eqref{eq:approx_sparsity1}: under the
assumption that many categories have the same effect, choosing the reference
category at random will often lead to sparsity.

However, such an assumption may be restrictive. Suppose, for instance, that we
want to nonparametrically control for a variable that measures age. Then a
plausible form of sparsity arises under an assumption that the age effect is a
step function, with some dividing line between young and old. If this threshold
is not close to either very young or very old, then a reference category
approach will not yield sparsity, but an appropriate re-expression of the column
space will. Or maybe there are three distinct coefficients, one for the young,
one for the middle-aged and one for the old, which again requires a different
specification of the fixed effects to induce sparsity. Such specifications are
less common in applied work. Arguably this is not because they are a priori less
plausible, but rather, it is well understood that for \ac{OLS}, they are all
equivalent to the reference category specification---so it suffices to consider
the reference category specification.

By contrast, when imposing a sparsity assumption on the coefficients, this
equivalence breaks down. To capture these more general sparsity patterns,
consider alternative specifications for the column space that take the form
\begin{equation}
  W_{i}=A_{0}Z_{i}, \quad A_{0}\in \mathbf{A}=\{A\in\mathbb{R}^{p\times
    p}\colon A_{ij}\in \{0,1\}, \quad\text{$A$ is full rank}\},
\label{eq:cat_regressor}
\end{equation}
where the $Z_{i}$ indicates the baseline categories (so each $Z_{i}$ is a column
of $I_{p}$). By construction, elements of $W_{i}$ are all equal to zero or one
in this specification, and since $A\in\mathbb{R}^{p\times p}$ is full rank, the
constant vector is part of the column space spanned by $W_{i}$.

Suppose with $A_{0}\in \mathbf{A}$, the coefficient vector on $W_{i}$ is sparse.
However, because the matrix $A_{0}$ that induces sparsity is unknown, we pick a
column space specification $\tilde{W}_{i}=\mathcal{A}Z_{i}$, where
$\mathcal{A}\in\mathbf{A}$ is selected at random by repeatedly drawing
$p\times p$ matrices with i.i.d.~Bernoulli$(q)$ entries, $0<q\leq 1/2$, until we
find one that is full rank. By Theorem A of \textcite{tikhomirov20}, the
probability of discarding a matrix in this process is vanishingly small, as it
is smaller than $(1-q+\varepsilon)^{p}$ for all $\varepsilon >0$ and large
enough $p$. This allows us to essentially ignore rank-deficient matrices,
yielding the following result.

\begin{theorem}\label{thm:categories}
  Suppose a single coefficient on $W_{i}=A_{0}Z_{i}$ is constant
  and non-zero, and the number of zeros $K$ in the corresponding row of
  $A_{0}$ satisfies $0<\lim_{n\rightarrow \infty}K/p<1$. If all baseline
  categories have population fractions of the same order, then the probability
  that the model with $\tilde{W}_{i}=\mathcal{A} Z_{i}$ satisfies
  \cref{eq:approx_sparsity1,eq:approx_sparsity2} is no larger than
  \begin{equation*}
    (1-q+\varepsilon)^{K}
  \end{equation*}
  for all $\varepsilon >0$ and large enough $p$.
\end{theorem}

\subsection{Offset in Hermite polynomial regression}\label{sec:set-herm-polyn}

Our last example involves a large number of technical regressors. Specifically,
suppose there is a scalar variable $z_{i}\sim\mathcal{N}(0,1)$
which we want to control for nonparametrically by a series regression. We
consider scaled Hermite polynomials
\begin{equation*}
  H_{j}(x)=(j!)^{-1/2}(-1)^{j}e^{x^{2}/2}\frac{d^{j}}{dx^{j}}e^{-x^{2}/2}, \quad
  j=0, \dotsc, p-1
\end{equation*}
which are particularly convenient for an underlying standard Gaussian variable,
since we have $E[\tilde{W}_{i}\tilde{W}_{i}^{\prime}] = I_{p}$ for
$\tilde{W}_{i, j}=H_{j-1}(z_{i})$.

Suppose further that if $z_{i}$ is shifted by a constant
$\lambda \in\mathbb{R}$, the expansion is extremely sparse, with only the
highest order term being non-zero, that is,
$Y_{i}=D_{i}\beta +H_{p-1}(z_{i}+\lambda)+U_{i}$. If we were to draw $\lambda $
uniformly on $[0,1]$, what is the probability that a researcher ignoring the
shift obtains an approximately sparse representation? The following result shows
that the probability will be bounded above by $\varepsilon/\log (p)$, for all
$\varepsilon > 0$ and large enough $p$.

\begin{theorem}\label{thm:Hermite}
  Suppose $\lambda =L/\log p$, $L>0$. If $L$ is fixed, then for
  $1\leq j\leq L\sqrt{p}/\log p$ and $p\geq \max \{(2L)^{2}, 6\}$,
  $\tilde{\gamma}_{p-j}^{2} \geq Ce^{j/2}$, where $C$ is an absolute constant,
  and \cref{eq:approx_sparsity1,eq:approx_sparsity2} fail. In contrast, if
  $L\rightarrow 0$, \cref{eq:approx_sparsity1,eq:approx_sparsity2} hold.
\end{theorem}

While the rate of $\log p$ in \Cref{thm:Hermite} is quite slow, note that when
sparsity fails, it does so dramatically in the sense that many coefficients
diverge. This is not surprising. It is well-known that assuming that the first
$s$ (say) terms of the polynomial expansion provide a good approximation to the
regression function amounts to imposing a smoothness assumption on it. In this
case, the exact basis approximation would matter little, and sparsity-based
methods would perform well regardless of the choice of $\lambda$. However, under
smoothness \acp{SBE} are not needed, since a low-order series expansion (i.e.\
only controlling for the first $s$ polynomial terms, with $s\ll p$) already
performs well.

In the current context with technical regressors, the appeal of \acp{SBE} is
that the assumption of sparsity they leverage is weaker than assuming
smoothness, since we allow for the best $s$-term approximation to include high
order polynomials. But if a good approximation to the regression function indeed
requires high-order terms, it must be very non-smooth---and for non-smooth
functions even small changes in the approximating basis imply wild swings in the
coefficients. In other words, while sparsity is indeed weaker than smoothness,
the particular class of non-smooth functions sparsity allows for is tightly
linked to the basis used. If one suspects non-smoothness, and one wants to avoid
running \ac{OLS} with many polynomial terms, one therefore needs a
substantive argument for why the particular basis chosen is likely to capture
it.

\subsection{Discussion}\label{sec:discussion}

The fact that the sparsity assumption depends on normalization choices in the
specification of $W$ is well-known in the theoretical literature. The three
stylized thought-experiments above quantify the extent of this dependence: the
vast majority of normalization choices do not induce an approximately sparse
representation. The results in \Cref{sec:empir-illustr} show that this lack of
invariance induces variation in \ac{SBE} estimates that is of the same order as
sampling uncertainty. This fragility appears underappreciated in the literature:
as \Cref{tab:t_lit1} shows, we were not able to ascertain the form of the
control matrix $W$ in \LRdesc\% of the papers surveyed, let alone find any
discussion of why the chosen specification for the control matrix should lead to
sparsity.\footnote{One paper mentioned that it included all dummy variables and
  did not pick a reference category; the paper did not discuss the implications
  of this choice for the restricted eigenvalue condition. The remaining papers
  in which the form of the control matrix could be inferred from the text did
  not feature categorical variables.}

A natural reaction to these results is to seek to modify \acp{SBE} in a way that
they adapt to different forms of sparse representations. For instance, it may be
possible to treat the offset $\lambda $ as another parameter to be estimated.
Similarly, a number of proposals have been developed that modify the lasso to
make it invariant to the choice of reference category, including the group lasso
\parencite{YuLi06}, or variants of fused lasso \parencite{tsrzk05}, such as
all-pairs penalties \parencite{BoRe09,GeTu10} or the SCOPE estimator
\parencite{StShTi21}. When combined with the debiasing techniques developed in
the literature on debiased lasso \parencite[e.g.][]{JaMo14,vbrd14,ZhZh14,bch14},
these approaches may yield algorithms that are less fragile.

However, it appears non-trivial to extend these modifications to the more
complicated designs encountered in practice, such as when interactions are
present, or to handle collinearity issues arising due to limited variability in
the controls. Both the BCH and the \textcite{ferrara22} applications of
\Cref{sec:empir-illustr} exhibit multicollinearity that goes beyond the issue of
choosing baseline categories.

An alternative approach could be to exhaustively consider all plausible
normalization choices, analogous to our analysis in \Cref{sec:empir-illustr},
and report the union of the confidence intervals. However, this is practically
difficult, as the number of possible normalizations is often very large.
Moreover, such an approach would only yield valid inference if at least one of
the considered choices induces a sparse representation.

To ensure that the latter condition holds, one needs a substantive argument. In
some instances, such arguments may necessitate further expanding the set of
normalization choices.
For instance, the grouped lasso approach in the references cited above require that the
baseline categories can be partitioned into a small number of types, with all
categories of the same type having (nearly) identical effect on the outcome
variable. In the context of controlling for profession dummies, say, this
amounts to an assumption that there are only a few profession types with
heterogeneous effects, which is a strong assumption. It may be more reasonable to
instead posit that professions are characterized by a combination of a small
number of latent skills, and these skills affect the outcome in an additive
manner. This then yields a different set of plausible normalizations that are
not captured by grouped lasso-type approaches.

In absence of substantive arguments, the sparsity assumption may well fail to
hold, even after considering a large set of potential normalization choices. We
present evidence for this in our applications in \Cref{sec:testing-sparsity}
below. Thus, developing fully data-driven methods to modify \acp{SBE} such that
they adapt to an appropriate sparse representation regardless of the context
seems infeasible.

Applications of \acp{SBE} require researchers to take a stand on why a
particular representation (or a set of representations) is sparse based on
domain-specific knowledge. This is analogous to using problem-specific arguments
to defend other substantive assumptions, such as assuming selection on
observables. But it is at odds with the purported appeal of \acp{SBE} to be
fully automatic, with the aim of reducing the researchers' degrees of freedom in
model choice.

While the sensitivity of \acp{SBE} to the control matrix specification may seem
similar to the usual issue that with flexible methods, implementation issues,
such as tuning parameter choice, can matter a lot in finite samples, it is
fundamentally different. %
Large sample validity of flexible methods can be justified under a wide range of
tuning parameter choices, including the choice of penalty in implementing
\acp{SBE}. By contrast, the specification of the control matrix affects validity
of the sparsity assumption, and thus large-sample validity of inference. It is a
substantive modeling choice, not merely an implementation detail. It needs to be
defended as such.

\section{Efficiency gains under sparsity}\label{sec:effic-gains-under}

One argument for using \acp{SBE} when $p$ is large, but still smaller than the
sample size, is that \ac{OLS} estimates are too noisy. In this section, we
quantify the potential efficiency gains of alternative estimators relative to
\ac{OLS}. We argue that potential gains are limited when the ratio $p/n$ is
small.

We consider the linear model given in \cref{eq:outcome}, and assume the
propensity score equation is also linear,
\begin{equation}\label{eq:p_score}
  D_{i}=W_{i}'\delta+\tilde{D}_{i}, \qquad E[\tilde{D}_{i}\mid W_{i}]=0.
\end{equation}
The assumption that the residuals $U_{i}$ and $\tilde{D}_{i}$ are conditionally
mean zero avoids technical complications; it could be replaced by an assumption
that restricts these conditional means to be small, so that the linear
specifications in \cref{eq:outcome,eq:p_score} could be thought of as
approximating non-linear conditional means \parencite[see, e.g.,][]{cjn18}.
Using the Frisch-Waugh-Lovell theorem, we can write the \ac{OLS} estimator of
$\beta$ as
\begin{equation*}
  \hat{\beta}_{OLS}:=\frac{\sum_{i=1}^{n}\ddot{D}_{i}Y_{i}}{\sum_{i=1}^{n}\ddot{D}_{i}^{2}},
\end{equation*}
where $\ddot{D}_{i}$ denotes the sample propensity score residual from
estimating~\cref{eq:p_score} by \ac{OLS}. Specifically, employing the usual
matrix notation, $\ddot{D}_{i}$ corresponds to the $i$th element of
$\ddot{D}:=(I-P)D$, where $P=WW^{+}$ is the projection matrix onto the space
spanned by the controls. Here $W^{+}$ denotes the pseudo-inverse (so that if $W$
is full rank, $P=W(W'W)^{-1}W'$).

To analyze $\hat{\beta}_{OLS}$, we need to impose more structure on the model in
\cref{eq:outcome,eq:p_score}.
\begin{assumption}\label{assumption:an_long} For some constants $\eta>0$ and
  $K\geq 1$:
  (i) $\{(U_{i}, \tilde{D}_{i})\}_{i=1}^{n}$ are independent across $i$
  conditional on $W$; (ii)
  $E[\abs{U_{i}}^{2+\eta}\mid D, W]+E[\abs{\tilde{D}_{i}}^{4}\mid W]\leq K$
  uniformly over $i$; (iii)
  $1/E[\tilde{D}_{i}^{2}\mid W]+1/E[U_{i}^{2}\mid D, W]\leq K$ uniformly over
  $i$; (iv) $\limsup_{n\to\infty}p/n<1$.
\end{assumption}
Parts (i)--(iii) of \Cref{assumption:an_long} are standard assumptions on
sampling and the regression errors, ensuring that the errors are neither too
thick-tailed nor degenerate. Part (iv) ensures that, when \ac{OLS} estimates
of~\eqref{eq:outcome} implicitly estimate the propensity score regression, we
do not overfit so much that we eliminate any variation in the sample residuals
$\ddot{D}_{i}$. This overfitting precludes using \ac{OLS} in settings with
$p>n$. However, as long as the limit of $p/n$ is strictly smaller than $1$, the
\ac{OLS} estimator remains asymptotically normal with the usual sandwich
expression for its standard error:

\begin{lemma}\label{theorem:asymptotic_normality}
  Consider the model in \cref{eq:outcome,eq:p_score}.
  Under~\Cref{assumption:an_long},
  \begin{equation*}
    \frac{\hat{\beta}_{OLS}-\beta}{
      \ss_{OLS}}
    \overset{d}{\to} \mathcal{N}(0,1), \qquad \ss_{OLS}^{2}
    =\frac{1}{(\ddot{D}'\ddot{D})^{2}}
    \sum_{i=1}^{n}\ddot{D}_{i}^{2}U_{i}^{2}.
  \end{equation*}
\end{lemma}

The fact that \ac{OLS} remains asymptotically normal in the regime $p\asymp n$
is not new---results similar to \Cref{theorem:asymptotic_normality} appeared
previously in, for instance, \textcite{cjn18} who build on earlier work by
\textcite{mammen93}. We state it here under a simpler set of assumptions to
allow us to consider its implications for potential efficiency gains. To this
end, observe that the asymptotic variance is in general larger than the
semiparametric efficiency bound unless $p/n\to 0$. In particular, when the
errors $U_{i}$ are homoskedastic, the efficient standard error is given by the
square root of
\begin{equation*}
  \ss^{2}_{*}=\frac{1}{(\tilde{D}'\tilde{D})^{2}}
    \sum_{i=1}^{n}\tilde{D}_{i}^{2}U_{i}^{2},
\end{equation*}
in the sense that the semiparametric efficiency bound is given by the
probability limit of $\ss_{*}^{2}/n$ \parencite[see, e.g.,][]{robinson88}.
Semiparametrically efficient estimators $\hat{\beta}^{*}$ of $\beta$ thus need
to satisfy
\begin{equation}\label{eq:efficient_if}
 \hat{\beta}^{*}-\beta=(\tilde{D}'\tilde{D})^{-1}\tilde{D}'U+o_{p}(n^{-1/2}),
\end{equation}
This lack of semiparametric efficiency is not a deficiency of \ac{OLS}: under
Gaussian errors, one-sided $t$-tests based on \ac{OLS} are uniformly most
powerful. Rather, it reflects the fact that when $p\asymp n$, achieving the
semiparametric efficiency bound requires some type of restriction on $\gamma$.
It turns out that restricting $\gamma$ to be sparse (see
\Cref{sec:approximate-sparsity}) is sufficient; indeed a number of \acp{SBE}
satisfy \cref{eq:efficient_if} \parencite[e.g.][]{bch14,JaMo14,vbrd14,ZhZh14}.

The standard error ratio $\ss_{*}/\ss_{OLS}$ thus represents the potential
efficiency gain from imposing sparsity. When $U_{i}$ is homoskedastic and
\Cref{assumption:an_long} holds, the ratio satisfies
\begin{equation*}
  \frac{\ss^{2}_{*}}{\ss^{2}_{OLS}}=
  (1-p/n)\frac{\ddot{D}'\ddot{D}/(n-p)}{\tilde{D}'\tilde{D}/n}(1+o_{p}(1))
  =
  (1-p/n)\kappa(1+o_{p}(1)), \qquad \kappa=
  \frac{E[\hat{\sigma}^{2}_{\tilde{D}}]}{\sigma^{2}_{\tilde{D}}},
\end{equation*}
where $\sigma^{2}_{\tilde{D}}$ is the variance of the error in the propensity
score regression, and $\hat{\sigma}^{2}_{\tilde{D}}=\ddot{D}'\ddot{D}/(n-p)$,
the mean squared error in the propensity score regression, can be thought of as
an estimator of $\sigma^{2}_{\tilde{D}}$. The ratio $\kappa$ measures the
relative bias of this estimator, with $\kappa=1$ when the estimator is unbiased.
This is the case when $\tilde{D}_{i}$ is homoskedastic or when the design is
balanced in that the leverages $P_{ii}$ are approximately equal. In this case,
the standard error ratio $\ss_{*}/\ss_{OLS}$ simplifies to a degrees-of-freedom
correction, $\sqrt{1-p/n}$, as noted previously in \textcite{cjn18et}. When
$\tilde{D}_{i}$ is heteroskedastic, $\hat{\sigma}^{2}_{\tilde{D}}$ may display
downward bias if the leverages $P_{ii}$ are positively correlated with the
conditional variances $E[\tilde{D}_{i}^{2}\mid W]$; this can be seen from
writing
$\kappa=E[\sum_{i}(1-P_{ii})E[\tilde{D}_{i}^{2}\mid W]]/
\sum_{i}(1-P_{ii})E[\tilde{D}_{i}^{2}]$. However, the correlation would need to
be substantial to induce downward bias that is sizable enough to allow for large
efficiency gains. When $p/n=0.2$ and $\kappa=0.9$, for instance, corresponding
to 10\% downward bias, the standard error ratio implies a
$1-\sqrt{0.8\cdot 0.9}=15.1\%$ reduction in standard error. Consequently, the
ratio $p/n$ needs to be much larger than $0.2$ to allow for sizable efficiency
gains under sparsity: at $\kappa=0.9$, we need $p/n\geq 6/16\approx 0.38$ to
reduce the standard errors by more than 25\%.

While these efficiency calculations rely on homoskedastic errors $U_{i}$,
similar conclusions are likely to hold under heteroskedasticity. The argument
for this is that, in general,
\begin{equation*}
  \frac{\ss_{*}}{\ss_{OLS}}=
  \frac{\ss_{*}/\ss_{*, \hom}}{\ss_{OLS}/\ss_{OLS, \hom}}
  \sqrt{(1-p/n)\kappa}\cdot (1+o_{p}(1)),
\end{equation*}
where $\ss_{*, \hom}=\sigma_{U}/\sqrt{\tilde{D}'\tilde{D}}$ and
$\ss_{OLS, \hom}=\sigma_{U}/\sqrt{\ddot{D}'\ddot{D}}$ are homoskedasticity-only
standard error formulas. The ratios $\ss_{*}/\ss_{*, \hom}$ and
$\ss_{OLS}/\ss_{OLS, \hom}$ measure the magnitude of the heteroskedasticity
correction on the standard errors. Heteroskedasticity would thus need to have a
large \emph{differential} impact on the standard errors of \ac{OLS} versus
$\hat{\beta}_{*}$ for the standard error ratio to deviate much from
$\sqrt{(1-p/n)\kappa}$. But magnitudes of heteroskedasticity corrections tend to
be modest in practice. Thus, while substantive efficiency gains are possible
when the ratio of $p/n$ is close to one, meaningful gains are unlikely when the
controls number 20\% or less of the sample size. In such cases---which,
according to our survey, cover the majority of \ac{SBE} applications---simply
running \ac{OLS} offers a robust method of inference at little efficiency loss.

\section{Testing sparsity}\label{sec:testing-sparsity}

In this section, we develop two tests of the sparsity assumption under our
maintained assumption that $p$ is smaller than, but proportional to the sample
size. We then apply these tests to the empirical examples studied in
\Cref{sec:empir-illustr}.

\subsection{Hausman test}\label{sec:hausman-test}

The first test we develop is a simple application of the idea popularized by
\textcite{hausman78} that if we have two estimators, one of which is more
efficient, but requires stronger assumptions for its validity, we can indirectly
test these stronger assumptions by checking whether the estimates are
statistically significantly different. The next lemma formalizes the result.

\begin{lemma}\label{theorem:hausman-test}
  Consider the model in \cref{eq:outcome,eq:p_score}. Suppose
  \Cref{assumption:an_long} holds, and that $\trace(P)/n\geq \tilde{K}$ a.s.\
  for some $\tilde{K}>0$. Then, for any estimator satisfying
  \cref{eq:efficient_if},
  \begin{equation*}
    \frac{\hat{\beta}_{OLS}-\hat{\beta}_{*}}{\ss_{H}}\overset{d}{\to}\mathcal{N}(0,1), \qquad
    \ss^{2}_{H}=\sum_{i=1}^{n}Z_{i}^{2}U_{i}^{2},
  \end{equation*}
  where
  $Z_{i}=\frac{\ddot{D}_{i}}{\ddot{D}'\ddot{D}}-
  \frac{\tilde{D}_{i}}{\tilde{D}'\tilde{D}}$.

  Additionally, suppose the regression functions admit the decomposition
  $W_{i}'\gamma=f(W_{i})+r_{\gamma}(W_{i})$ and
  $W_{i}'\delta=g(W_{i})+r_{\delta}(W_{i})$, where the remainder terms
  satisfy (i)
  $\frac{1}{n}\sum_{i=1}^{n}E[r_{\delta}(W_{i})^{2}+r_{\gamma}(W_{i})^{2}]\to
  0$; (ii)
  $\frac{1}{n}\sum_{i=1}^{n}E[r_{\gamma}(W_{i})^{2}r_{\delta}(W_{i})^{2}]$ is
  bounded; and (iii)
  $n\sum_{i=1}^{n}
  (Z_{i}^{2}U_{i}^{2}-\tilde{z}_{i}^{2}(U_{i}+r_{\gamma}(W_{i}))^{2})=o_{p}(1)$,
  with
  $\tilde{z}_{i}=\frac{\ddot{D}_{i}}{\ddot{D}'\ddot{D}}-\frac{\tilde{D}_{i}+r_{\delta}(W_{i})}{
    \sum_{i=1}^{n}(\tilde{D}_{i}+r_{\delta}(W_{i}))^{2}}$. Suppose also that for
  some estimates $\hat{U}=\hat{U}(Y, D, W)$ and $\hat{D}=\hat{D}(W, D)$ (iv)
  $\max_{i}\abs{\hat{U}_{i}-U_{i}-r_{\gamma}(W_{i})}+\max_{i}\abs{\hat{D}_{i}-D_{i}-r_{\delta}(W_{i})}=o_{p}(1)$.
  Then the same conclusion holds with $\ss_{H}^{2}$ replaced by
  $\hat{\ss}^{2}_{H}=\sum_{i=1}^{n}\hat{Z}_{i}^{2}\hat{U}_{i}^{2}$, where
  $\hat{Z}_{i}=\frac{\ddot{D}_{i}}{\ddot{D}'\ddot{D}}-
  \frac{\hat{D}_{i}}{\hat{D}'\hat{D}}$.
\end{lemma}

The second part of \Cref{theorem:hausman-test} allows us to construct a simple
plug-in estimator of the standard error $s_{H}$ based on lasso or post-lasso
residuals. In particular, when the regression functions admit a sparse
approximation in the sense of \cref{eq:approx_sparsity1,eq:approx_sparsity2},
the additional condition (i) in the second part of \Cref{theorem:hausman-test}
will hold with $f$ and $g$ given by the best sparse approximations to
$W_{i}'\gamma$ and $W_{i}'\delta$, and condition (iv) will hold for the lasso or
the post-lasso residuals. Conditions (ii) and (iii) are high-level conditions
ensuring that if we include the approximation errors $r_{\delta}$ and
$r_{\gamma}$ in the definition of the residuals, replacing $U_{i}$ with
$U_{i}+r_{\gamma}(W_{i})$, and $\tilde{D}_{i}$ with
$\tilde{D}_{i}+r_{\delta}(W_{i})$, this has negligible impact on the standard
error $s_{H}$ in large samples; it is similar to the condition ASTE (P) (v) in
\textcite{bch14}. We note that \Cref{theorem:hausman-test} only imposes very
weak conditions on the control matrix $W$, allowing it to be reduced-rank, so
long as the rank is proportional to $p$, and allowing the rows to be dependent,
and not identically distributed.

Sometimes, \acp{SBE} are used as a ``robustness check'' alongside a main
specification based on \ac{OLS}. \Cref{theorem:hausman-test} shows that such
practice is in fact the opposite of a robustness check: if the two estimates are
not close to one another this indicates failure of the sparsity assumption
rather than lack of robustness in the \ac{OLS} estimates. When $U_{i}$ is
homoskedastic, the Hausman standard error may be written as
\begin{equation*}
  \ss^{2}_{H}=(\ss^{2}_{OLS}-\ss^{2}_{*})(1+o_{p}(1))=
  \ss^{2}_{*}\left(\frac{1}{(1-p/n)\kappa}-1\right)(1+o_{p}(1)),
\end{equation*}
so that when the efficiency gain $\sqrt{(1-p/n)\kappa}$ is small, the two
estimates need to be tightly coupled, within a fraction of the \ac{SBE} standard
error.

\subsection{Residual test}\label{sec:residual-test}

Our second approach to testing sparsity is based on the idea that if the
identities $\mathcal{S}^{*}\subseteq \{1,\dotsc, p\}$ of the controls that give
the best sparse approximation to the regression function are known, testing
sparsity is equivalent to testing that the coefficients on the remaining
controls are small, which can be gauged using a conventional $F$-statistic.
Although the identities $\mathcal{S}^{*}$ are unknown in practice, it turns out
that under the null hypothesis of sparsity, the residuals from the infeasible
short regression that only includes controls in $\mathcal{S}^{*}$ are
sufficiently well approximated by residuals from a lasso regression, so that the
hypothesis can be tested by comparing the lasso and \ac{OLS} residuals.

To make the result precise, consider a linear regression
\begin{equation}\label{eq:f-regression}
  \yy_{i}=X_{i}'\alpha+\epsilon_{i}, \quad E[\epsilon_{i}\mid X_{i}]=0,\qquad
  i=1,\dotsc, n,
\end{equation}
with $\epsilon_{i}$ independent across $i$, conditional on the regressors, and
$p:=\dim(X_{i})<n$. In our setup, \cref{eq:f-regression} may correspond to one
of three regressions. For testing sparsity of $\gamma$ in \cref{eq:outcome},
which all \acp{SBE} require, we can set $\yy_{i}=Y_{i}$ and
$X_{i}=(D_{i}, W_{i}')'$. For testing sparsity in the propensity score
regression, which is needed, for instance, for the validity of the post-double
selection estimator of \textcite{bch14}, we can set $\yy_{i}=D_{i}$ and
$X_{i}=W_{i}$. Sparsity in both regressions can be tested jointly in the reduced
form regression of $\yy_{i}=Y_{i}$ onto $X_{i}=W_{i}$.

We wish to test the assumption that the regression function $X_{i}'\alpha$
admits a sparse approximation. To work out the implications of this hypothesis,
we introduce some additional notation. For a subset
$\mathcal{S}\subseteq\{1,\dotsc, p\}$ of the regressors, let $X_{S}$ denote the
submatrix of $X$ that drops the columns corresponding to the complement of
$\mathcal{S}$, let $P_{\mathcal{S}}=X_{\mathcal{S}}X_{\mathcal{S}}^{+}$ denote
the projection matrix associated with $X_{\mathcal{S}}$, and let $P=XX^{+}$
denote the full projection matrix. In a slight departure from
\Cref{sec:sparse-repr-are}, we gauge the quality of the approximation
conditional on $X$, so that the approximation error from only using the
regressors $X_{\mathcal{S}}$ is given by $(I-P_{\mathcal{S}})X\alpha$, the
residual from projecting $X\alpha$ onto $X_{\mathcal{S}}$. The assumption that
$X\alpha$ is sparse can then be stated as:
\begin{assumption}\label{assumption:approx_sparsity}
  There exists a subset $\mathcal{S}^{*}\subseteq \{1,\dotsc, p\}$ with
  cardinality $s$, such that
  $\norm{(I-P_{\mathcal{S}^{*}})X\alpha}_{2}^{2}=O_{p}(s)$ and
  $s\log (p)/\sqrt{p}\to 0$.
\end{assumption}

If we knew the identity of the subset $\mathcal{S}^{*}$, and we also assumed
that the sparsity was exact, so that the $O_{p}(s)$ term
in~\Cref{assumption:approx_sparsity} was zero, then a natural way of testing
\Cref{assumption:approx_sparsity} would be to compare the restricted and
unrestricted sum of squared residuals,
\begin{equation*}
  \mathcal{F}=\yy'(I-P_{\mathcal{S}^{*}})\yy-\yy'(I-P)\yy.
\end{equation*}
The statistic $\mathcal{F}$ is the numerator of the homoskedastic $F$-statistic.
While the $F$-statistic critical values are only valid under homoskedasticity,
we can leverage the fact that when $p\to\infty$, as is the case under our
asymptotics, $\mathcal{F}$ is asymptotically normal after centering and scaling
to derive a critical value that is robust to heteroskedasticity. Furthermore, if
we have an estimator $\hat{\alpha}$ such that $\yy-X\hat{\alpha}$ approximates
the residuals $(I-P_{\mathcal{S}^{*}})\yy$ from the infeasible short regression
sufficiently well, we can replace the infeasible sum of squared residuals
$\yy'(I-P_{\mathcal{S}^{*}})\yy$ in $\mathcal{F}$ by
$\norm{\yy-X\hat{\alpha}}_{2}^{2}$ to derive a feasible version of this test.
Finally, it turns out that weakening exact sparsity to approximate sparsity in
the sense of \Cref{assumption:approx_sparsity} does not impact the null rejection
probability of the test in large samples. The next lemma formalizes these
arguments.

\begin{lemma}\label{theorem:non-normal-errors_epe}
  Suppose \Cref{assumption:approx_sparsity} holds and that, for some $K\geq 1$,
  (a) $E[\epsilon_{i}^{4}\mid X]\leq K$; (b) $\max_{i}P_{ii}<1-1/K$; (c)
  $\trace(P)/n\geq 1/K$; and (d) $E[\epsilon_{i}^{2}\mid X]\geq 1/K$ a.s. Then, as
  $n\to\infty$,
  \begin{equation}\label{eq:infeasible}
    \frac{\mathcal{F}-\sum_{i}\epsilon_{i}^{2}P_{ii}}{
      \sqrt{2\sum_{i\neq j}\epsilon_{i}^{2}\epsilon_{j}^{2}P_{ij}^{2}}
    }\overset{d}{\to}\mathcal{N}(0,1).
  \end{equation}
  Furthermore, suppose that for some estimator $\hat{\alpha}$, (i)
  $\norm{X\hat{\alpha}-X\alpha}_{2}\preceq_{p} \sqrt{s\log(p)}$; (ii)
  $\norm{\hat{\alpha}-(X_{\mathcal{S}^{*}}'X_{\mathcal{S}^{*}})^{-1}
    X_{\mathcal{S}^{*}}'X\alpha}_{1}\preceq_{p} s\sqrt{\log(p)/n}$; and, in
  addition, (iii) $\max_{ij}\abs{X_{ij}}s\sqrt{\log(p)/n}=o_{p}(1)$; and (iv)
  $\norm{\sum_{i}\epsilon_{i}P_{ii}X_{i}}_{\infty}+
  \norm{\sum_{i}\epsilon_{i}X_{i}}_{\infty}\preceq_{p}\sqrt{n\log(p)}$. Then,
  letting $\hat{\epsilon}_{i}=\yy_{i}-X_{i}'\hat{\alpha}$,
  \begin{equation*}
    \frac{\norm{\hat{\epsilon}}_{2}^{2}-\yy'(I-P)\yy-\sum_{i}\hat{\epsilon}_{i}^{2}P_{ii}}{
      \sqrt{2\sum_{i\neq j}\hat{\epsilon}_{i}^{2}\hat{\epsilon}_{j}^{2}P_{ij}^{2}}
    }\overset{d}{\to}\mathcal{N}(0,1).
  \end{equation*}
\end{lemma}

Conditions (a) and (d) are standard assumptions on the regression errors, while
conditions (b) and (c) impose weak restrictions on the design matrix. Condition
(b) bounds leverage away from one, which implies \Cref{assumption:an_long}~(iv),
ensuring the sample residuals $\ddot{D}_{i}$ display non-zero variation.
Condition (c) is analogous to the condition imposed in
\Cref{theorem:hausman-test}. Conditions (i) and (ii) in the second part of the
\namecref{theorem:non-normal-errors_epe} are standard mean-squared error and
$\ell_{1}$ rate conditions that hold for the lasso and post-lasso estimators
\parencite{brt09,BeCh13}. Conditions (iii) and (iv) are tail restrictions on the
covariates and residuals analogous to those in \textcite{bch14}.

\Cref{theorem:non-normal-errors_epe} implies that we can test the sparsity
assumption by calculating the lasso or post-lasso residuals
$\hat{\epsilon}_{i}$, and then checking whether the residual sum of squares of
the lasso is comparable to that of the \ac{OLS} residual sum of squares. If the
difference satisfies
\begin{equation*}
        \norm{\hat{\epsilon}}_{2}^{2}-\yy'(I-P)\yy
\geq z_{1-\alpha}\sqrt{2\sum_{i\neq j}\hat{\epsilon}_{i}^{2}\hat{\epsilon}_{j}^{2}P_{ij}^{2}}
+\sum_{i}\hat{\epsilon}_{i}^{2}P_{ii}
\end{equation*}
where $z_{1-\alpha}$ is the $1-\alpha$ quantile of the standard normal
distribution (1.645 for $\alpha=5\%$ level test), we reject
\Cref{assumption:approx_sparsity}.

In general, testing hypotheses about $\alpha$ when $\lim_{n\to\infty}p/n>0$ can
be quite involved, since plugging in regression residuals into asymptotic
variance expressions leads to bias \parencite[see, e.g.,][]{kss20,AnSo23,cjn18}.
As shown in the proof of \Cref{theorem:non-normal-errors_epe} (see
\cref{eq:replace_variance}), we avoid such difficulties here by virtue of the
fact that under the assumptions of the lemma, the lasso residuals
$\hat{\epsilon}_{i}$ are sufficiently accurate. As a result, the variance part
of the test statistic, the denominator in \cref{eq:infeasible}, can be
consistently estimated under the null hypothesis by plugging in lasso residuals.

\subsection{Empirical tests of sparsity}\label{sec:empir-tests-spars}

While \Cref{theorem:hausman-test,theorem:non-normal-errors_epe} show that the
two tests we develop control size in large samples, analysis of their power
properties is challenging, as this involves analysis of the lasso or post-lasso
estimators without imposing sparsity, and, to our knowledge, there are no
existing results in the literature one could leverage. One may therefore be
concerned that the tests are not powerful in empirically relevant settings.

To evaluate this concern, we now apply the tests developed in
\Cref{sec:hausman-test,sec:residual-test} to the empirical illustrations
considered in \Cref{sec:empir-illustr}. For each specification, we consider the
Hausman test, which compares \ac{OLS} with \ac{SBE}, as well two versions of the
residual test. The first version estimates the outcome regression
in~\cref{eq:outcome} using the post-lasso, and compares the residuals to those
based on \ac{OLS}. The second version compares post-lasso and \ac{OLS} residuals
from estimating the propensity score regression in~\cref{eq:p_score}.

\Cref{tab:t3} reports the results. Column 1 applies the test to the original
specification in each paper. We see that for 6 of the 7 outcomes, at least
one of the tests rejects the assumption of sparsity.

\begin{table}
  \renewcommand*{\arraystretch}{1.2}
  \begin{threeparttable}
    \caption{$p$-values (in percentages) for tests of the sparsity assumption under different normalizations of
      the control matrix.}\label{tab:t3}
    \small
  \begin{tabular*}{0.95\linewidth}{@{\extracolsep{\fill}}ll@{} S@{}S@{}S@{}S@{}S@{}S@{}S@{}S@{}S@{}}
    & & \multicolumn{1}{c}{Repl.}
    & \multicolumn{2}{c}{Collinearity}
    & \multicolumn{2}{c}{Powers}     & \multicolumn{2}{c}{Category sums}
    & \multicolumn{2}{c}{Offset} \\
    \cmidrule(rl){4-5}\cmidrule(rl){6-7}\cmidrule(rl){8-9}\cmidrule(rl){10-11}
    Outcome & Test  & \multicolumn{1}{c}{(1)} & \multicolumn{1}{c}{(2)} & \multicolumn{1}{c}{(3)} & \multicolumn{1}{c}{(4)} & \multicolumn{1}{c}{(5)}
& \multicolumn{1}{c}{(6)} & \multicolumn{1}{c}{(7)}
& \multicolumn{1}{c}{(8)} & \multicolumn{1}{c}{(9)}
\\
    \midrule
      \csname @@input\endcsname ./table3.tex
    \bottomrule
  \end{tabular*}
  \begin{tablenotes}
  \item{}\footnotesize\emph{Notes}: See notes to \Cref{tab:t1} for description
    of the specifications. H\@: Hausman, OR\@: residual test based on outcome
    regression, PR\@: residual test based on propensity score regression. Col.~1
    reports $p$-values under the original specification for the BCH,
    \textcite{ferrara22}, and \textcite{enke20} studies. Cols.~2--3 report the
    range of $p$-values obtained under alternative ways of dropping collinear
    columns of the control matrix. Cols.~4--5, 6--7, and 8--9 report,
    respectively, the $p$-value range under alternative normalizations of
    the controls prior to taking powers and interactions, when categorical
    variables are expressed as indicators for different subsets, and when powers
    are constructed using Hermite polynomials, with an offset $\lambda$ ranging
    between $-1$ and $1$.
  \end{tablenotes}
\end{threeparttable}
\end{table}

To check that these results are not driven by finite-sample size distortions of
these tests, \Cref{sec:monte-carlo} conducts Monte Carlo simulations based on
these applications; in these simulations, the size stays close to or below the
nominal level.

One response to these findings is to seek alternative normalizations of the
control matrix that are consistent with the sparsity assumptions. The remaining
columns in \Cref{tab:t3} report the range of $p$-values under the four different
normalizations we considered in \Cref{sec:empir-illustr}. The table shows that
in these applications, we were unable to find a normalization for five of the
seven outcomes where at least one test did not reject. This is in spite of the
large number of alternative specifications that these normalizations generate.

\section{Conclusions}\label{sec:conclusion}

We have argued, using empirical evidence and theoretical arguments, that
\acp{SBE} display a lack of robustness to the specification of the control
matrix. In the three applications we have examined, the range of variation in
the \ac{SBE} estimates under equally plausible alternative specifications is of
the same order of magnitude as sampling uncertainty. Two reasons underlie this
fragility. First, whether a small number of covariates can account for the bulk
of the confounding depends on the particular specification of the control
matrix: even if the sparsity assumption holds under a particular way of
expressing the column space of the controls, most alternative plausible
normalizations do not admit a sparse approximation. Second, it may be the case
no control matrix in a large class of normalizations admits a sparse
approximation.

What should a practitioner take away from these results? We have argued that
unless $p$ is comparable to $n$, the potential efficiency gains of \acp{SBE}
over \ac{OLS} are limited. Simply reporting \ac{OLS} with standard errors that
are robust to the presence of many controls
\parencite[e.g.,][]{cjn18,jochmans22,DoSu18,dadamo19} will deliver credible
inference at little efficiency loss. When $p$ is comparable to $n$, \ac{OLS}
becomes too noisy to be useful, and infeasible when the covariate dimension
exceeds the sample size. In such cases informative inference requires
restricting the control vector, either by making sparsity assumptions, limiting
the magnitude of the control coefficients \parencite[e.g.,][]{LiMu21,akk20}, or
combining both assumptions \parencite[e.g.,][]{chl17}. If the high
dimensionality of the controls arises because one wishes to flexibly control for
a moderate-dimensional control vector, another possibility is to impose
smoothness restrictions on the outcome and propensity score equations and use
debiased machine learning \parencite[e.g.,][]{ccddhnr18}. All these alternatives
involve substantive modeling restrictions, and as such, they need to be
discussed and defended on substantive grounds, analogous to the discussion of
other key assumptions, such as selection on observables. In particular,
researchers who opt to leverage \acp{SBE} need to explain why sparsity should
plausibly hold under the chosen specification of the control matrix, and not
leave normalization choices to statistical software. The sparsity tests
developed in this paper can serve as a complement to these arguments, provided
they are used as a model specification check rather than a pretest.

We have focused on the sparsity assumption and \acp{SBE} because these
estimators are used frequently, and their theory is well-developed. However,
many other modern machine learning methods likewise lack invariance to linear
reparametrization of the control matrix. When these methods are used for
prediction, this lack of invariance is less important, for two reasons. First,
one is typically interested in average performance over many predictions, and
the overall prediction performance may be robust even if individual predictions
are sensitive to normalizations. Second, one can gauge the performance of a
given procedure directly using a test sample. When we incorporate these methods
into econometric models, however, we are typically interested in inference on a
single causal effect, and test sample benchmarking is unavailable. Understanding
more generally when a lack of invariance leads to fragility is an interesting
area for future research.

\begin{appendices}

\renewcommand{\figurename}{Appendix Figure}
\setcounter{figure}{0}
\renewcommand{\tablename}{Appendix Table}
\setcounter{table}{0}

  \section{Auxiliary results}\label{sec:auxiliary-results}
  For the next two results, we consider a quadratic form $\psi'H\psi$, where $H$
  is an orthogonal projection matrix with rank bounded by $r$. Conditional on
  some $\sigma$-algebra $\mathcal{Z}_{n}$, the elements of $\psi$ are
  independent and mean zero, and $H$ is non-random. We will prove a law of large
  numbers and a central limit theorem for $\psi'H\psi$.

\begin{lemma}\label{theorem:quadratic_form}
  Suppose that uniformly over $i$,
  $E[\abs{\psi_{i}}^{2+\eta}\mid \mathcal{Z}_{n}]\leq K$ for some $K>0$ and some
  $\eta\in[0,2]$. Then
  \begin{equation*}
    \psi'H\psi=E[\psi'H\psi\mid\mathcal{Z}_{n}]+O_{p}(r^{1/2}+r^{2/(2+\eta)}).
  \end{equation*}
\end{lemma}
\begin{proof}
  Write
  \begin{equation*}
    \psi'H \psi-E[\psi'H\psi\mid\mathcal{Z}_{n}]=
    \sum_{i}(\psi_{i}^{2}-E[\psi_{i}^{2}\mid\mathcal{Z}_{n}])H_{ii}+2\sum_{i<j}\psi_{i}\psi_{j}H_{ij}=:T_{1}+T_{2}.
  \end{equation*}
  By iterated expectations, and the inequality of \textcite{vBEs65},
  \begin{equation*}
    E\abs*{T_{1}}^{1+\eta/2}
    \leq 2E\sum_{i}\abs{\psi_{i}^{2}-E[\psi_{i}^{2}\mid\mathcal{Z}_{n}]}^{1+\eta/2}\abs{H_{ii}}^{1+\eta/2}
    \leq 2Kr,
  \end{equation*}
  so by Markov's inequality, $T_{1}=O_{p}(r^{2/(2+\eta)})$. The term $T_{2}$ is
  mean zero with variance
  \begin{equation*}
    4\sum_{i<j}E[\psi_{i}^{2}\psi_{j}^{2}H_{ij}^{2}]\preceq r,
  \end{equation*}
  so by Markov's inequality, $T_{2}=O_{p}(r^{1/2})$.
\end{proof}

\begin{lemma}\label{sec:qf_clt}
  Let
  $\omega^{2}=2\sum_{i\neq j}H_{ij}^{2}E[\psi_{i}^{2}\psi_{j}^{2}\mid
  \mathcal{Z}_{n}]$. Suppose that for some constant $K\geq 1$, (a)
  $r/\omega^{2}\leq K$, and (b) $E[\psi_{i}^{4}\mid\mathcal{Z}_{n}]\leq K$ a.s.
  Then, as $r\to\infty$,
  \begin{equation*}
    \frac{\psi'H\psi-\sum_{i}H_{ii}\psi_{i}^{2}}{\omega}\overset{d}{\to}\mathcal{N}(0,1).
  \end{equation*}
\end{lemma}
\begin{proof}
  Write
  $(\psi'H\psi-\sum_{i}H_{ii}\psi_{i}^{2})/\omega=\sum_{i}\mathcal{Y}_{i}$,
  where $\mathcal{Y}_{i}=\frac{2}{\omega}\psi_{i}\sum_{j=1}^{i-1}\psi_{j}H_{ij}$
  is a martingale difference array with respect to the filtration
  $\mathcal{F}_{i}=\sigma(\psi_{1}, \dotsc, \psi_{i}, \mathcal{Z}_{n})$. By the
  martingale central limit theorem, it suffices to verify the Lyapunov condition
  $\sum_{i}E[\mathcal{Y}_{i}^{4}]\to 0$, and convergence of the conditional
  variance, $\sum_{i}E[\mathcal{Y}_{i}^{2}\mid \mathcal{F}_{i-1}]=1+o_{p}(1)$.

  The Lyapunov condition follows from the bound
  \begin{multline*}
    \sum_{i}E[\mathcal{Y}_{i}^{4}]
    \leq
    \frac{2^{4}}{\omega^{4}}E\left[\psi^{4}_{i}\cdot 3\sum_{j=1}^{i-1}\psi_{j}^{2}H_{ij}^{2}
      \sum_{k=1}^{i-1}\psi_{k}^{2}H_{ik}^{2}\right]\\
    \leq \sum_{i}E
    \frac{48 K^{3}}{\omega^{4}}\sum_{j=1}^{i-1}\sum_{k=1}^{i-1}
    H_{ij}^{2}
    H_{ik}^{2}
    \leq
    \frac{48 K^{3}}{\omega^{4}}r
    \leq
    \frac{48 K^{5}}{r}\to 0,
  \end{multline*}
  where the last inequality uses condition (a).

  To show convergence of the conditional variance, decompose
  \begin{multline*}
    \sum_{i}E[\mathcal{Y}_{i}^{2}\mid \mathcal{F}_{i-1}]-1
    =
    \frac{4}{\omega^{2}}\sum_{i=1}^{n}\sum_{j=1}^{i-1}E[\psi_{i}^{2}\mid\mathcal{Z}_{n}](\psi_{j}^{2}
    -E[\psi_{j}^{2}\mid\mathcal{Z}_{n}])H_{ij}^{2}\\
    + \frac{8}{\omega^{2}}\sum_{i=1}^{n}E[\psi_{i}^{2}\mid\mathcal{Z}_{n}]\sum_{j=1}^{i-1}\psi_{j}H_{ij}
    \sum_{k=1}^{j-1}\psi_{k}H_{ik}=: T_{1} + T_{2}.
  \end{multline*}
  We have
  \begin{equation*}
    E[T_{1}^{2}\mid \mathcal{Z}_{n}]  \leq \frac{16}{\omega^{4}}\sum_{j=1}^{n}K^{3}\left(\sum_{i=j+1}^{n}H_{ij}^{2}\right)^{2}
    \leq \frac{16 K^{5}}{r},
  \end{equation*}
  so that $T_{1}=o_{p}(1)$ as $r\to\infty$ by Markov's inequality. It remains to
  show that $T_{2}=o_{p}(1)$. Let $\norm{A}_{F}$ denote the Frobenius norm of a
  matrix, and let $L_{ij}=\1{i>j}H_{ij}$. The second moment of $T_{2}$ can be
  bounded as
  \begin{multline*}
    E[T_{2}^{2}\mid \mathcal{Z}_{n}]=
    \frac{64}{\omega^{4}}\sum_{k>j}
    E[\psi_{j}^{2}\psi_{k}^{2}\mid\mathcal{Z}_{n}]\left(\sum_{i=1}^{n}
      E[\psi_{i}^{2}\mid\mathcal{Z}_{n}]L_{ij}H_{ik}\right)^{2}\\
    \leq
    \frac{64 K^{4}}{r^{2}}\sum_{j, k}
    \left(\sum_{i=1}^{n} E[\psi_{i}^{2}\mid\mathcal{Z}_{n}]L_{ij}H_{ik}\right)^{2}\\
    =\frac{64 K^{4}}{r^{2}} \norm{H \diag(E[\psi_{i}^{2}\mid\mathcal{Z}_{n}])L}_{F}^{2}
    \preceq
    \frac{\norm{L}_{F}^{2}}{r^{2}}\leq
    \frac{\norm{H}_{F}^{2}}{r^{2}}=\frac{1}{r},
  \end{multline*}
  where the first inequality uses conditions (a) and (b), and the second
  inequality applies the inequality
  $\norm{AL}_{F}^{2}\leq \lambda_{\max}(A'A)\norm{L}_{F}^{2}$ twice. The claim
  then follows by Markov's inequality.
\end{proof}
\begin{lemma}\label{lem:BernoulliMat}
  Let $\mathcal{B}_{p}^{K}$ be a $K\times p$ matrix with i.i.d.~Bernoulli$(q)$
  entries, $0<q\leq1/2$, and assume $0<\lim_{p\rightarrow \infty}K/p\leq 1$. For
  any $\varepsilon >0$, there exist $C_{1}, C_{2}>0$ only depending on
  $\varepsilon$ and $q$ such that
  \begin{equation*}
    P\left(\min_{\norm{v}_{2}\leq \lfloor C_{1}p\rfloor, \norm{v}_{2}=1}
      \norm{\mathcal{B}_{p}^{K}v}_{2}^{2}\leq C_{2}p\right) < (1-q+\varepsilon)^{K}
  \end{equation*}
  for all large enough $p$.
\end{lemma}

\begin{proof}
  The proof follows from the same arguments as Proposition 3.6 of
  \textcite{tikhomirov20}. There are only two instances in the proof where
  $p-K\neq 1$ matters: first, in the second to last displayed inequality on
  Tikhomirov's page 601, the second inequality now invokes Lemma 3.5 with $m=K$,
  so in his notation, the two $n-1$ terms are replaced by $K$ (where his $n$
  corresponds to our $p$). Second, in the last displayed inequality on page 601,
  the $n-1$ term is again replaced by $K$. Since
  $\lim_{n\rightarrow \infty}K/p=c_{1}>0$, the second inequality now still
  holds for all $|\sum x_{i}|$ that are larger by a factor of $2/\sqrt{c_{1}}$,
  at least for all large enough $n$. This only affects the constant $C$ in
  Tikhomirov's first displayed inequality on page 602, and the result follows as
  in Tikhomirov's proof.
\end{proof}

\begin{lemma}\label{theorem:F_approx_clust}
  Suppose that \Cref{assumption:approx_sparsity} holds, and that condition (a),
  and conditions (i), (ii), and (iv) of \Cref{theorem:non-normal-errors_epe}
  hold. Then
  \begin{equation*}
    \norm{\yy-X\hat{\alpha}}_{2}^{2}-
    \yy'(I-P_{\mathcal{S}^{*}})\yy=O_{p}(s\log(p)).
  \end{equation*}
\end{lemma}
\begin{proof}[Proof of \Cref{theorem:F_approx_clust}]
  Letting $r=(I-P_{\mathcal{S}^{*}})X\alpha$, and
  $\tilde{\alpha}=\hat{\alpha}-(X_{\mathcal{S}^{*}}'X_{\mathcal{S}^{*}})^{-1}X_{\mathcal{S}^{*}}'X\alpha$,
  we may write,
  \begin{equation*}
    \begin{split}
      \norm{\yy-X\hat{\alpha}}_{2}^{2}-
      \yy'(I-P_{\mathcal{S}^{*}})\yy
      &=
        \norm{X\alpha-X\hat{\alpha}}_{2}^{2}+2\epsilon'(X\alpha-X\hat{\alpha})+\epsilon'\epsilon
        -\norm{(I-P_{\mathcal{S}^{*}})\epsilon+r}_{2}^{2}\\
      &= \norm{X\alpha-X\hat{\alpha}}_{2}^{2}-2\epsilon'X\tilde{\alpha}
        +\epsilon'P_{\mathcal{S}^{*}}\epsilon-\norm{r}_{2}^{2}.
    \end{split}
  \end{equation*}
  Hence, by \cref{eq:eps_Ps_eps}, \Cref{assumption:approx_sparsity}, and
  condition (i) of \Cref{theorem:non-normal-errors_epe},
  \begin{equation*}
    \abs*{\norm{\yy-X\hat{\alpha}}_{2}^{2}-
      \yy'(I-P_{\mathcal{S}^{*}})\yy}
    \preceq_{p} s\log(p)+\norm{\epsilon'X}_{\infty}\norm{\tilde{\alpha}}_{1}\preceq_{p}
    s\log(p),
  \end{equation*}
  where the second inequality uses conditions (ii) and (iv) of
  \Cref{theorem:non-normal-errors_epe}.
\end{proof}

\section{Proofs}
\allowdisplaybreaks%

\begin{proof}[Proof of \Cref{thm:rotation}]
  Let $A_{p}$ be the probability of the event that the model with regressor
  $\mathcal{R}W_{i}$ satisfies \cref{eq:approx_sparsity1,eq:approx_sparsity2},
  and let
  \begin{equation*}
    r_{s, \mathcal{R}}=\min_{\norm{v}_{0}\leq s}E[(W_{i}'\mathcal{R}'\mathcal{R}\gamma
    -W_{i}^{\prime}\mathcal{R}^{\prime}v)^{2}\mid \mathcal{R}]
    =\min_{\norm{v}_{0}\leq s}(\mathcal{R}\gamma -v)^{\prime}\mathcal{R}
    E[W_{i}W_{i}^{\prime}]\mathcal{R}'(\mathcal{R}\gamma -v)
  \end{equation*}
  denote the mean square approximation error under the best $s$-sparse
  approximation in this model. Using the assumption that $E[W_{i}W_{i}']$ has
  eigenvalues bounded from below, \cref{eq:rot_dist}, and the assumption that
  $\norm{\gamma}_{2}$ is bounded from below, it follows that for some $K>0$
  large enough,
\begin{equation*}
    r_{s, \mathcal{R}}
    \geq\frac{1}{K}\min_{\norm{v}_{0}\leq s}
    \norm{\mathcal{R}\gamma -v}_{2}^{2} \geq\frac{1}{K^{2}}
    \frac{\sum_{j=1}^{p-s}\mathcal{Z}_{j:p}^{2}}{\sum_{j=1}^{p}\mathcal{Z}_{j}^{2}}
\end{equation*}
where $\mathcal{Z}_{j:p}^{2}$ are the order statistics
$\mathcal{Z}_{1:p}^{2} \leq \mathcal{Z}_{2:p}^{2}\leq \cdots \leq
\mathcal{Z}_{p:p}^{2}$ of the sample $\{\mathcal{Z}_{j}^{2}\}_{j=1}^{p}$, and
$\mathcal{Z}_{j}$ are i.i.d.~standard normal.

If \cref{eq:approx_sparsity1,eq:approx_sparsity2} hold, then for $p$ large
enough, $r_{s, \mathcal{R}}\leq \eta :=K^{-2} s/p$ for
$s=\lfloor \sqrt{p}/\log (p)\rfloor $. Hence, by the union bound, for $p$ large
enough,
\begin{equation*}
  \begin{split}
    A_{p}& \leq P\left(r_{s, \mathcal{R}}\leq \eta \right)\\
         & \leq \binom{p}{s}P\left(\frac{1}{K^{2}}
           \frac{\sum_{j=1}^{p-s}\mathcal{Z}_{j}^{2}}{
           \sum_{j=1}^{p}\mathcal{Z}_{j}^{2}}\leq \eta \right)\\
         & =\binom{p}{s}P\left(
           \frac{\sum_{j=1}^{p-s}\mathcal{Z}_{j}^{2}/(p-s)}{\sum_{j=p-s}^{p}\mathcal{Z}_{j}^{2}/s}\leq
           \frac{s}{p-s}\frac{K^{2}\eta}{1-K^{2}\eta}\right)\\
         & =\binom{p}{s}I_{K^{2}\eta}(p/2-s/2, s/2),
\end{split}
\end{equation*}
where the last line uses the fact that an $F$-distribution with $d_{1}$ and
$d_{2}$ degrees of freedom has
c.d.f.~$I_{d_{1}x/(d_{1}x+d_{2})}(d_{1}/2,d_{2}/2)$. Here
$I_{x}(a;b)=\int_{0}^{x}t^{a-1}(1-t)^{b-1}dt/B(a, b)$ is the regularized
incomplete beta function, and $B(a, b)$ is the beta function. Since
$\int_{0}^{x}t^{a-1}(1-t)^{b-1}dt\leq \int_{0}^{x}t^{a-1}dt=x^{a}/a$,
\begin{multline}\label{eq:apbound}
  A_{p} \leq \binom{p}{s}\frac{1}{B(p/2-s/2, s/2)}\frac{(K^{2}\eta)^{p/2-s/2}}{p/2-s/2} \\
  \leq e^{3s/2}(p/s)^{3s/2-1}(K^{2}\eta)^{p/2-s/2}
  =e^{3s/2}(p/s)^{2s-1-p/2}\leq
  e^{3s/2}(\sqrt{p})^{2s-1-p/2},
\end{multline}
where the second line uses the identity
$\binom{p}{s}=\frac{p}{s(p-s)}\frac{1}{B(s, p-s)}$ and a double application of
the beta function bound $\frac{p}{s(p-s)}\frac{1}{B(s, p-s)} \leq (ep/s)^{s}$
for all $s\geq 1$, and the last inequality uses the definition of $s$ and holds
for $p$ large enough so that $2s-1-p/2$ is negative. The beta function bound
follows from the generalization of Stirling's inequality to gamma functions due
to \textcite[Theorem 5]{gordon94},
\begin{equation}\label{eq:stirling_gamma}
  \Gamma (t)=\sqrt{2\pi}t^{t-1/2}e^{-t}R_{t}, \qquad R_{t}\in [1, \sqrt{2}]
\end{equation}
for all $t\geq 1/2$.\footnote{Specifically,
  \begin{multline*}
    \frac{p}{s(p-s)}\frac{1}{B(s, p-s)}=\frac{p\Gamma (p)}{s\Gamma (s)(p-s)\Gamma (p-s)}=\frac{R_{p}}{\sqrt{2\pi}R_{s}R_{p-s}}\frac{(p/s)^{s}}{s^{1/2}}\left(1+\frac{s}{p-s}\right) ^{p-s+1/2}\\\leq \frac{1}{\sqrt{\pi}}\sqrt{\frac{p}{s(p-s)}}(ep/s)^{s}\leq (ep/s)^{s},
  \end{multline*}
  where the second equality uses \cref{eq:stirling_gamma}, the first inequality
  uses $1+x\leq e^{x}$ and the bounds on $R_{t}$ in \cref{eq:stirling_gamma},
  and the last inequality uses $p/(s(p-s))\leq 2$.} Taking logs of
\cref{eq:apbound} then yields the upper bound
\begin{equation*}
\log (A_{p})\leq \tfrac{3}{2}s-(p/4-s+1/2)\log (p)\asymp -\frac{p}{4}\log (p).
\end{equation*}

To derive a lower bound for $A_{p}$, observe that
$r_{s, \mathcal{R}}\leq
K^{2}\sum_{j=1}^{p-s}\mathcal{Z}_{j:p}^{2}/\sum_{j=1}^{p}\mathcal{Z}_{j}^{2}$
for a large enough $K$. Since the probability of an approximately sparse
representation is lower bounded by the probability that
$r_{s, \mathcal{R}} \leq K^{2}s/p$ for
$s=\lfloor \sqrt{p}/\log (p)^{2}\rfloor $, we have
\begin{equation*}
\begin{split}
  A_{p}& \geq P\left(\frac{\sum_{j=1}^{p-s}\mathcal{Z}_{j}^{2}}{
         \sum_{j=1}^{p}\mathcal{Z}_{j}^{2}}\leq s/p\right) = I_{s/p}(p/2-s/2, s/2) \\
       & \geq \frac{(1-s/p)^{s/2-1}(s/p)^{p/2-s/2}}{(p/2-s/2) B(p/2-s/2, s/2)}\\
       & \geq
         \frac{1}{2\sqrt{2\pi}}
         \frac{(s/p)^{p/2-s+1}}{
         (s/2)^{1/2} (1-s/p)^{p/2-s+3/2}} \geq \frac{1}{2\sqrt{\pi}}\frac{(s/p)^{p/2-s+1}}{\sqrt{s}},
\end{split}
\end{equation*}
where the second line uses the bound
$\int_{0}^{x}t^{a-1}(1-t)^{b-1}dt\geq (1-x)^{b-1}\int_{0}^{x}t^{a-1}dt =
(1-x)^{b-1}x^{a}/a$, and the third line uses the inequality
$\frac{p}{s(p-s)B(s, p-s)}%
\geq\frac{1}{2\sqrt{2\pi}} \frac{(p/s)^{s}}{s^{1/2} (1-s/p)^{p-s+1/2}} $ that
follows from \cref{eq:stirling_gamma}. Hence,
$\log (A_{p})\succeq -\frac{p}{4}\log(p)$, as claimed.
\end{proof}

\begin{proof}[Proof of \Cref{thm:categories}]
  Let $s=\lfloor \sqrt{p}/\log (p)\rfloor$. If
  \cref{eq:approx_sparsity1,eq:approx_sparsity2} hold, then for $p$ large
  enough,
  $\min_{\norm{v}_{0}\leq s}E[(Z_{i}^{\prime}A_{0}' \gamma
  -Z_{i}^{\prime}\mathcal{A}' v)^{2}\mid \mathcal{A}]\leq s/p$. It hence
  suffices to show that given $\varepsilon >0$, for all large enough $p$,
  \begin{equation*}
    P\left(\min_{\norm{v}_{0}\leq s}E[(Z_{i}^{\prime}A_{0}'\gamma
      -Z_{i}^{\prime}\mathcal{A}'v)^{2}\mid \mathcal{A}]
      \leq \frac{s}{p}\right) <(1-q+\varepsilon
    )^{K}.
  \end{equation*}
  Let $\psi =A_{0}'\gamma$, a $p\times 1$ vector with $K$ elements equal to zero
  and $p-K$ elements equal to $\norm{\gamma}_{2}$ (which is equal to the single
  non-zero element of $\gamma $). By assumption, $p E[Z_{i}Z_{i}^{\prime}]$ is a
  diagonal matrix with all diagonal elements bounded below by some constant
  $C_{0}>0$. Thus, it suffices to show that
  \begin{equation}\label{eq:catproof}
    P\left(\min_{\norm{v}_{0}\leq s}\norm{\psi -\mathcal{A}'v}_{2}^{2}\leq s/C_{0}\right)
    <(1-q+\varepsilon)^{K}
  \end{equation}
  for all $p$ large enough.

  Let $\mathcal{B}$ be a $p\times p$ matrix with i.i.d.~Bernoulli$(q)$ entries,
  and without loss of generality, assume $\mathcal{A}=\mathcal{B}$ if
  $\mathcal{B}$ is non-singular. By Theorem A of \textcite{tikhomirov20}, given
  any $\varepsilon_{0}>0$, the probability of $\mathcal{B}$ being singular is
  smaller than $(1-q+\varepsilon_{0})^{p}$ for all large enough $p$. Applying
  this result with $\varepsilon_{0}=\varepsilon/2$, and using the law of total
  probability, we can bound the left-hand side by
  \begin{equation*}
    \begin{split}
      P\left(\min_{\norm{v}_{0}\leq s}\norm{\psi -\mathcal{A}'v}_{2}^{2}\leq s/C_{0}\right)
      & \leq P\left(\text{$\mathcal{B}$ is singular}\right)
        + P\left(\min_{\norm{v}_{0}\leq s}\norm{\psi -\mathcal{B}'v}_{2}^{2}\leq s/C_{0}\right)\\
      & \leq (1-q+\varepsilon/2)^{p}
        +P\left(\min_{\norm{v}_{0}\leq s}\norm{\psi -\mathcal{B}'v}_{2}^{2}\leq s/C_{0}\right).
    \end{split}
  \end{equation*}
  Since $2(1-q+\varepsilon/2)^{p}\leq (1-q+\varepsilon)^{p}$ for
  $p\geq \log(2)/\log(\tfrac{1-q+\varepsilon}{1-q+\varepsilon/2})$, to show
  \cref{eq:catproof}, suffices to show that the second term in the above display
  is bounded by $(1-q+\varepsilon/2)^{p}$.

With $\eta =\norm{\gamma}_{2}/(2\sqrt{s})$, by the union bound,
\begin{multline}\label{eq:catproof2}
  P\left(\min_{\norm{v}_{0}\leq s}\norm{\psi -\mathcal{B}v}_{2}^{2}\leq s/C_{0}\right)
  \leq
  P\left(\min_{\norm{v}_{0}\leq s, \norm{v}_{2}<\eta}\norm{\psi -\mathcal{B}v}_{2}^{2}\leq
    s/C_{0}\right)+ \\
  P\left(\min_{\norm{v}_{0}\leq s, \norm{v}_{2}\geq \eta}
    \norm{\psi -\mathcal{B}v}_{2}^{2}\leq s/C_{0}\right).
\end{multline}
Note that all elements of $\mathcal{B}$ are either zero or one. By the
Cauchy-Schwarz inequality, if $\norm{v}_{0}\leq s$ and $\norm{v}_{2}\leq \eta $,
all elements of $\mathcal{B}v$ are bounded above by
$\tfrac{1}{2}\norm{\gamma}_{2}$, and hence
$\norm{\psi -\mathcal{B}v}_{2}^{2}\geq \tfrac{1}{4}\norm{\gamma}_{2}^{2}(p-K)$
almost surely. The first probability on the right-hand side of
\cref{eq:catproof2} is thus equal to zero for all large enough $p$.

Let $\mathcal{B}_{p}^{K}$ be the $K\times p$ matrix that collects the $K$ rows
of $\mathcal{B}$ where the corresponding element in $\psi $ is zero. Then for
any $C_{1}>0$,
\begin{equation*}
  \min_{\norm{v}_{0}\leq s, \norm{v}_{2}>\eta}\norm{\psi -\mathcal{B}v}_{2}^{2}
  \geq
  \min_{\norm{v}_{0}\leq s, \norm{v}_{2}\geq \eta}
    \norm{\mathcal{B}_{p}^{K}v}_{2}^{2}
  \geq \eta ^{2}\min_{\norm{v}_{0}\leq \lfloor C_{1}p\rfloor, \norm{v}_{2}=1}
          \norm{\mathcal{B}_{p}^{K}v}_{2}^{2}
\end{equation*}
almost surely, where the last inequality holds for all $p$ large enough so that
$\lfloor C_{1}p\rfloor \geq s$. Thus, for all $C_{2}>0$ and large enough $p$,
\begin{multline*}
  P\left(\min_{\norm{v}_{0}\leq s, \norm{v}_{2}\geq \eta}\norm{\psi -\mathcal{B}v}_{2}^{2}\leq
    s/C_{0}\right)
  \leq P\left(\min_{\norm{v}_{2}\leq \lfloor C_{1}p\rfloor, \norm{v}_{2}=1}
    \norm{\mathcal{B}_{p}^{K}v}_{2}^{2}\leq s/\eta ^{2}\cdot 1/C_{0}\right) \\
  \leq P\left(\min_{\norm{v}_{2}\leq \lfloor C_{1}p\rfloor, \norm{v}_{2}=1}
    \norm{\mathcal{B}_{p}^{K}v}_{2}^{2}\leq C_{2}p\right)
\end{multline*}
since $s/\eta ^{2}\asymp p/(\ln p)^{2}$. The result now follows from choosing
$C_{1}$ and $C_{2}$ from \Cref{lem:BernoulliMat}, with $\varepsilon/2$ playing the
role of $\varepsilon$.
\end{proof}

\begin{proof}[Proof of \Cref{thm:Hermite}]
  Let $p_{0}=p-1$. Using the identity
  \begin{equation*}
    \sqrt{j!} H_{j}(x+y)= \sum_{k=0}^{j}\binom{j}{k}x^{j-k}\sqrt{k!} H_{k}(y),
  \end{equation*}
  we find
\begin{equation*}
  H_{p_{0}}(z_{i}+\lambda)=\sum_{j=0}^{p_{0}}\binom{p_{0}}{j}\sqrt{j!/p_{0}!}
  \lambda ^{p_{0}-j}\tilde{W}_{i, j+1} = \sum_{k=0}^{p_{0}} \tilde{\gamma}_{p-k}
  \tilde{W}_{i, p-k},
\end{equation*}
where
$\tilde{\gamma}_{p-k}=\binom{p_{0}}{p_{0}-k}\sqrt{(p_{0}-k)!/p_{0}!}
\lambda^{k}$.

Applying the Stirling formula \parencite{robbins55}
\begin{equation*}
\log (n!)=\frac{1}{2}\log(2\pi) + (n+1/2)\log (n)-n+R_{n}, \qquad R_{n}\in \left[1/(12n+1),1/12n\right]
\end{equation*}
we obtain, for $k\geq 1$,
\begin{equation}\label{eq:coef2}
\begin{split}
  \log (\tilde{\gamma}_{p-k}^{2})
  & =\log \left(\frac{p_{0}!\lambda ^{2k}}{k!^{2}(p_{0}-k)!}\right) \\
  & =(p_{0}-k+1/2)\log \left(\frac{p_{0}}{p_{0}-k}\right) +k-\log (k)+2k\log
    \left(\frac{\sqrt{p_{0}}\lambda}{k}\right) +R_{p_{0}, k}^{\prime},
\end{split}
\end{equation}
where $R_{p_{0}, k}^{\prime}=R_{p_{0}}-2R_{k}-R_{p_{0}-k}-\log(2\pi)$ is bounded
above and below by a constant, since $R_{n}\in [0, 1/12]$. Now,
$\log (p_{0}/(p_{0}-k))=\log (1+k/(p_{0}-k))\geq k/p_{0}$ since
$\log (1+x)\geq x/(1+x)$ for $x\geq -1$. Hence, for some constant $C$, and for
$1\leq k\leq L\sqrt{p_{0}}/\log (p_{0})$
\begin{equation*}
\begin{split}
  \log (\tilde{\gamma}_{p-k}^{2})
  & \geq k\left[2+\frac{1/2-k}{p_{0}}-\log
    (k)/k\right] +2k\log \left(\frac{\sqrt{p_{0}}L}{k\log (p_{0})}\right) +C \\
  & \geq k\left[2+\frac{1/2-L\sqrt{p_{0}}/\log (p_{0})}{p_{0}}-\log (k)/k\right]
    + 2k\log \left(1\right) + C \\
  & \geq k\left[1-\frac{L}{\log (p_{0})\sqrt{p_{0}}}\right] +C,
\end{split}
\end{equation*}
where the last inequality uses $\log (k)/k\leq 1$. The expression in brackets is
bounded below by $1/2$ for $p\geq \max \{6,(2L)^{2}\}$, which yields the result
for $L$ fixed.

To prove the result when $L\to 0$, \cref{eq:coef2} implies, for some constant
$C$,
\begin{equation*}
\begin{split}
  \log (\tilde{\gamma}_{p-k}^{2})
  & \leq (p_{0}-k+1/2)\log\left(1+\frac{k}{p_{0}-k}\right)
    +k-\log (k)+2k\log \left(\frac{\sqrt{p_{0}}\lambda}{k}\right) +C \\
  & \leq 2k+\frac{1}{2}\frac{k}{p_{0}-k}
    + 2k\log \left(\frac{\sqrt{p_{0}}\lambda}{k}\right)
    + C\leq k\left[3+2\log \left(\frac{\sqrt{p_{0}}\lambda}{k}\right) \right]
    +C,
\end{split}
\end{equation*}
where the second line uses $\log (1+x)\leq x$. Since for
$k\geq \sqrt{p_{0}}\lambda e^{2}$ the expression in square brackets is smaller
than $-1$, if we approximate the regression function using the last
$s=\lceil e^{2}\max \{L/\log (p_{0}),1/\log (p_{0})^{2}\}\sqrt{p_{0}} \rceil$
regressors, we make an approximation error that is bounded by
\begin{equation*}
\sum_{k\geq s}\tilde{\gamma}_{p-k}^{2}\leq C\sum_{k\geq s}e^{-k}\preceq
e^{-s}\preceq e^{-e^{2}\sqrt{p}/\log (p)^{2}},
\end{equation*}
which is of lower order than $s/p$.
\end{proof}

\begin{proof}[Proof of \Cref{theorem:asymptotic_normality}]
  Write
  \begin{equation*}
    \hat{\beta}_{OLS}-\beta=\sum_{i=1}^{n}\eta_{i},
    \qquad \eta_{i}=\frac{1}{\ddot{D}'\ddot{D}} \ddot{D}_{i}U_{i}.
  \end{equation*}
  Conditional on $(D, W)$, $\eta_{i}$ are independent, with variance
  $\omega_{OLS}^{2}:=\var(\sum_{i}\eta_{i}\mid D,
  W)=\frac{\sum_{i}\ddot{D}_{i}^{2}E[U_{i}^{2}\mid D,
    W]}{(\ddot{D}'\ddot{D})^{2}}$. By a conditional version of the
  Lindeberg-Feller theorem \parencite[Theorem 1]{bulinski17}, it therefore
  suffices to verify a conditional version of the Lyapunov condition
  \begin{equation*}
    \sum_{i=1}^{n}
    \frac{E\left[\abs{\eta_{i}}^{2+\eta}\mid D, W\right]}{\omega_{OLS}^{2+\eta}}=o_{p}(1),
  \end{equation*}
  and that
  \begin{equation}\label{eq:cond_variance}
    \ss^{2}_{OLS}/\omega^{2}_{OLS}=1+o_{p}(1).
 \end{equation}
  To show the Lyapunov condition, substitute in the definition of $\eta_{i}$ and
  $\omega_{OLS}$ to write the left-hand side as
  \begin{multline*}
    \sum_{i=1}^{n} \frac{E[\abs{U_{i}}^{2+\eta}\mid D, W]\abs{\ddot{D}_{i}}^{2+\eta}}{
      (\sum_{j=1}^{n}\ddot{D}_{j}^{2} E[U_{j}^{2}\mid D, W])^{1+\eta/2}}\\
    \leq
    K^{2+\eta/2} \frac{\sum_{i=1}^{n}\abs{\ddot{D}_{i}}^{2+\eta}}{
      (\sum_{j=1}^{n}\ddot{D}_{j}^{2})^{1+\eta/2}}
    \leq
    K^{2+\eta/2} \frac{\max_{i'}\abs{\ddot{D}_{i'}}^{\eta}
      \sum_{i=1}^{n}\ddot{D}_{i}^{2}}{
      (\sum_{j=1}^{n}\ddot{D}_{j}^{2})^{1+\eta/2}}
    =
    K^{2+\eta/2} \left(\frac{\max_{i}\ddot{D}^{2}_{i}/n}{
        \frac{1}{n}\sum_{j=1}^{n}\ddot{D}_{j}^{2}}\right)^{\eta/2},
  \end{multline*}
  where the first inequality uses \Cref{assumption:an_long} (ii) and (iii).
  Therefore, to verify the Lyapunov condition, it suffices to show that
  \begin{equation}\label{eq:max_ddotd}
    \max_{i}\abs{\ddot{D}_{i}}=O_{p}(n^{1/4}),
  \end{equation}
  and that
  \begin{equation}
    \label{eq:mean_ddotd}
    \frac{1}{n}\sum_{i=1}^{n}\ddot{D}_{i}^{2}\asymp_{p} 1.
  \end{equation}
  Now, by~\Cref{theorem:quadratic_form} and \Cref{assumption:an_long} (ii) and
  (iii),
  $\frac{1}{n}\sum_{i=1}^{n}\ddot{D}_{i}^{2}=\frac{1}{n}\sum_{i}M_{ii}E[\tilde{D}_{i}^{2}\mid
  W]+o_{p}(1)$, $M_{ii}=(I-P)_{ii}$, with the first term bounded between
  $(1-p/n)/K$ and $K$. Therefore, \cref{eq:mean_ddotd} follows by
  \Cref{assumption:an_long} (iv). Furthermore, by the union bound and Markov's
  inequality
  \begin{equation*}
    P(\max_{i}\abs{\ddot{D}_{i}}/n^{1/4}\geq \epsilon)
    \leq \sum_{i=1}^{n}P(\abs{\ddot{D}_{i}}\geq n^{1/4}\epsilon)
    \leq\sum_{i=1}^{n} \frac{E[\ddot{D}_{i}^{4}]}{n\epsilon^{4}}
    \leq\frac{1}{\epsilon^{4}}4K,
  \end{equation*}
  which implies \cref{eq:max_ddotd}. Here the last inequality follows from the
  bound
  \begin{equation}\label{eq:ddot4}
    E[\ddot{D}_{i}^{4}\mid W]=
    \sum_{j}M_{ij}^{4}E[\tilde{D}_{j}^{4}\mid W]
    +3\sum_{j\neq k}M_{ij}^{2} M_{ik}^{2} E[\tilde{D}_{j}^{2}\tilde{D}_{k}^{2}\mid W]
    \leq
    4K\sum_{j, k}M_{ij}^{2} M_{ik}^{2}=4K M_{ii}^{2}.
  \end{equation}

  It remains to verify~\cref{eq:cond_variance}. Write
  \begin{equation*}
    \frac{\ss_{OLS}^{2}}{\omega_{OLS}^{2}}-1=
    \frac{\frac{1}{n}\sum_{i=1}^{n}\ddot{D}_{i}^{2}(U_{i}^{2}-E[U_{i}^{2}\mid
      D, W])}{\frac{1}{n}\sum_{i=1}^{n}\ddot{D}_{i}^{2}
      E[U_{i}^{2}\mid D, W]}
  \end{equation*}
  The denominator satisfies
  $\frac{1}{n}\sum_{i=1}^{n}\ddot{D}_{i}^{2} E[U_{i}^{2}\mid D, W]\geq
  \frac{1}{n}\sum_{i=1}^{n}\ddot{D}_{i}^{2}/K\asymp_{p}1$
  by~\cref{eq:mean_ddotd}. The result therefore follows if we can show that the
  numerator is of the order $o_{p}(1)$. This follows from the fact that for any
  variable $\mathcal{U}_{i}$, that, conditional on $(D, W)$ has mean zero and
 uniformly bounded $1+\eta/2$ moments,
  \begin{equation}\label{eq:ddotDU}
    \frac{1}{n}\sum_{i}\ddot{D}_{i}^{2}\mathcal{U}_{i}=o_{p}(1).
  \end{equation}
  In turn, \cref{eq:ddotDU} follows by Markov's inequality, and the bound
  \begin{multline*}
    E\abs*{\frac{1}{n}\sum_{i=1}^{n}\ddot{D}_{i}^{2}\mathcal{U}_{i}}^{1+\eta/2}\preceq
    \frac{1}{n^{1+\eta/2}}
    \sum_{i=1}^{n}E\abs{\ddot{D}_{i}}^{2+\eta}\leq
    \frac{1}{n^{1+\eta/2}} \sum_{i=1}^{n}E[\ddot{D}_{i}^{4}]^{(2+\eta)/4}\\
    \leq \frac{1}{n^{1+\eta/2}}
    E\sum_{i=1}^{n}\left(4KM_{ii}^{2}\right)^{(2+\eta)/4}\to 0
  \end{multline*}
  Here the first inequality uses iterated expectations, and the inequality of
  \textcite{vBEs65}, the second inequality uses Jensen's inequality, the third
  uses \cref{eq:ddot4}, and the final limit uses $M_{ii}^{2}\leq 1$.
\end{proof}

\begin{proof}[Proof of~\Cref{theorem:hausman-test}]
  Observe first that
  \begin{equation}\label{eq:z_mean}
    \begin{split}
      n\sum_{i=1}^{n}Z_{i}^{2}
      =&n\sum_{i=1}^{n}\left(\frac{\ddot{D}_{i}^{2}}{(\ddot{D}'\ddot{D})^{2}}+
      \frac{\tilde{D}_{i}^{2}}{(\tilde{D}'\tilde{D})^{2}}-\frac{2\ddot{D}_{i}\tilde{D}_{i}}{(\ddot{D}'\ddot{D})(\tilde{D}'\tilde{D})}\right)
    =n\frac{\tilde{D}'\tilde{D}+\ddot{D}'\ddot{D}-2\ddot{D}'\tilde{D}}{\tilde{D}'\tilde{D}\cdot\ddot{D}'\ddot{D}}\\
       & =n\frac{\tilde{D}'\tilde{D}+\tilde{D}'(I-P)\tilde{D}-2\tilde{D}'(I-P)\tilde{D}}{\tilde{D}'\tilde{D}\cdot\ddot{D}'\ddot{D}}
    =\frac{\tilde{D}'P\tilde{D}/n}{\tilde{D}'\tilde{D}/n\cdot \ddot{D}'\ddot{D}/n}\\
       & =\frac{\frac{1}{n}\sum_{i}P_{ii}E[\tilde{D}_{i}^{2}\mid W]+O_{p}(p^{1/2}/n)}{\tilde{D}'\tilde{D}/n
      \cdot \ddot{D}'\ddot{D}/n}
    \asymp_{p}\frac{\trace(W)}{n},
\end{split}
  \end{equation}
  where the first two lines follow by algebra, the equality on the third line
  follows by \Cref{theorem:quadratic_form}, and the last step follows by the law
  of large numbers, \cref{eq:mean_ddotd}, and \Cref{assumption:an_long} (ii) and
  (iii). Hence, letting
  $\omega_{H}^{2}=\sum_{i=1}^{n}Z_{i}^{2}E[U_{i}^{2}\mid D, W]$, it follows from
  \Cref{assumption:an_long} (iii), and the fact that $\ss^{2}_{*}=O_{p}(1/n)$
  that
  \begin{equation*}
    \frac{\ss_{*}^{2}}{\omega_{H}^{2}}=
    \frac{O_{p}(1)
    }{n\sum_{i=1}^{n}Z_{i}^{2}E[U_{i}^{2}\mid D, W]}=O_{p}(n/\trace(P)).
  \end{equation*}
  Therefore, under \cref{eq:efficient_if}, since
  $\liminf_{n\to\infty}n/\trace(P)$ is bounded,
  \begin{equation*}
   \frac{\hat{\beta}_{OLS}-\hat{\beta}^{*}}{\omega_{H}}
    =
   \frac{1}{\omega_{H}}\sum_{i=1}^{n}Z_{i}U_{i}+o_{p}(\ss_{*}/\omega_{H})=
   \frac{1}{\omega_{H}}\sum_{i=1}^{n}Z_{i}U_{i}+o_{p}(1).
  \end{equation*}
  By a conditional version of the Lindeberg-Feller theorem \parencite[Theorem
  1]{bulinski17}, the first term is asymptotically standard normal if a
  conditional version of the
  Lyapunov condition holds,
  \begin{equation}\label{eq:lyap_con}
        \sum_{i=1}^{n}\frac{\abs{Z_{i}}^{2+\eta}E[
      \abs{U_{i}}^{2+\eta}\mid D, W]}{\omega_{H}^{2+\eta}}=o_{p}(1).
  \end{equation}
  By~\Cref{assumption:an_long} (ii) and (iii) and \cref{eq:z_mean}, the
  left-hand side is bounded above by
  \begin{equation*}
    K^{2+\eta/2}\left(\frac{n\max_{i}Z_{i}^{2}}{n\sum_{i}Z_{i}^{2}}\right)^{\eta/2}
    \preceq_{p} (n^{2}/\trace(P)\cdot \max_{i}Z_{i}^{2})^{\eta/2}.
  \end{equation*}
  The claim in \cref{eq:lyap_con} then follows from the bound
  \begin{equation*}
    n\max_{i}Z_{i}^{2}\leq 2\frac{\max_{i}\ddot{D}_{i}^{2}/n}{(\ddot{D}'\ddot{D}/n)^{2}}+
    2\frac{\max_{i}\tilde{D}_{i}^{2}/n}{(\tilde{D}'\tilde{D}/n)^{2}}.
  \end{equation*}
  The first term is $O_{p}(n^{-1/2})$ by~\cref{eq:max_ddotd,eq:mean_ddotd}. The
  second term is also $O_{p}(n^{-1/2})$, since
  $\tilde{D}_{i}'\tilde{D}_{i}/n\asymp_{p} 1$ by the law or large numbers and
  \Cref{assumption:an_long} (ii) and (iii), and since
  \begin{equation}\label{eq:max_dtilde}
    \max_{i}\abs{\tilde{D}_{i}}= O_{p}(n^{1/4}).
  \end{equation}
  Specifically, \cref{eq:max_dtilde} holds since by the union bound and Markov's inequality,
  \begin{equation*}
    P(\max_{i}\abs{\tilde{D}}_{i}/n^{1/4}>\epsilon)
    \leq \sum_{i=1}^{n} P(\abs{\tilde{D}}_{i}/n^{1/4}>\epsilon)
    \leq
    \sum_{i=1}^{n} \frac{E[\tilde{D}_{i}^{4}]}{n\epsilon^{4}}
    \leq \frac{K}{\epsilon^{4}}.
  \end{equation*}
  To establish the first claim of \Cref{theorem:hausman-test}, it now suffices
  to show that $\omega_{H}/\ss_{H}=1+o_{p}(1)$. This follows from writing
  \begin{equation*}
    \frac{\ss_{H}^{2}}{\omega_{H}^{2}}-1
    =\frac{n\sum_{i}Z_{i}^{2}(U^{2}_{i}-E[U_{i}^{2}\mid D, W])}{n\sum_{i}Z_{i}^{2}E[U_{i}^{2}\mid D, W]}.
  \end{equation*}
  The denominator is of the order $p/n$ by \cref{eq:z_mean} and
  \Cref{assumption:an_long} (ii) and (iii). To show that the numerator is
  $o_{p}(1)$, let $\mathcal{U}_{i}=U_{i}^{2}-E[U_{i}^{2}\mid D, W]$, and decompose
  it as
  \begin{multline*}
    n\sum_{i}Z_{i}^{2}\mathcal{U}_{i}=
    \frac{\frac{1}{n}\sum_{i}
        \ddot{D}_{i}^{2}\mathcal{U}_{i}}{(\ddot{D}'\ddot{D}/n)^{2}}+
      \frac{\frac{1}{n}\sum_{i}\tilde{D}_{i}^{2}\mathcal{U}_{i}}{(\tilde{D}'\tilde{D}/n)^{2}}
      -2\frac{\frac{1}{n}\sum_{i}\ddot{D}_{i}\tilde{D}_{i}\mathcal{U}_{i}}{(\ddot{D}'\ddot{D}/n)(\tilde{D}'\tilde{D}/n)}\\
    =
    O_{p}(1)\frac{1}{n}\sum_{i}
      \ddot{D}_{i}^{2}\mathcal{U}_{i}+
    O_{p}(1)\frac{1}{n}\sum_{i}\tilde{D}_{i}^{2}\mathcal{U}_{i}
    +O_{p}(1)\frac{1}{n}\sum_{i}\ddot{D}_{i}\tilde{D}_{i}\mathcal{U}_{i}=o_{p}(1).
  \end{multline*}
  Here the second equality uses~\cref{eq:mean_ddotd}, and the last inequality
  uses~\cref{eq:ddotDU}, the law of large numbers, and the result that
  \begin{equation}\label{eq:ddottildeDU}
    \frac{1}{n}\sum_{i}\ddot{D}_{i}\tilde{D}_{i}\mathcal{U}_{i}=o_{p}(1),
  \end{equation}
  which follows by Markov's inequality and the bound
  \begin{multline*}
    E\abs*{\frac{1}{n}\sum_{i}\ddot{D}_{i}\tilde{D}_{i}\mathcal{U}_{i}}^{1+\eta/2}
    \preceq \frac{1}{n^{1+\eta/2}}
    \sum_{i}E\abs{\ddot{D}_{i}\tilde{D}_{i}}^{1+\eta/2}
    \leq
    \frac{1}{n^{1+\eta/2}}
    \sum_{i}(E\ddot{D}_{i}^{2}\tilde{D}_{i}^{2})^{(2+\eta)/4}\\
    =
    \frac{1}{n^{1+\eta/2}}
    \sum_{i}(E \sum_{j}M_{ij}^{2}\tilde{D}_{j}^{2}\tilde{D}_{i}^{2})^{(2+\eta)/4}
    \leq
    \frac{1}{n^{1+\eta/2}}
    \sum_{i}(M_{ii}K)^{(2+\eta)/4}\to 0.
  \end{multline*}
  Here the first inequality uses iterated expectations, and the inequality of
  \textcite{vBEs65}, the second inequality uses Jensen's inequality, the third
  uses \Cref{assumption:an_long} (ii), and the final limit uses $M_{ii}\leq 1$.

  To prove the second claim in \Cref{theorem:hausman-test}, let
  ${r}_{\gamma, i}=r_{\gamma}(W_{i})$, ${r}_{\delta, i}=r_{\delta}(W_{i})$,
  $\tilde{u}_{i}={r}_{\gamma, i}+U_{i}$ and $\tilde{d}_{i}=r_{\delta,
    i}+\tilde{D}_{i}$ so that we may write
  \begin{equation}\label{eq:var_cons2}
    \abs*{\frac{\hat{\ss}_{H}^{2}}{\omega_{H}^{2}}-1}
    \preceq_{p}\abs*{n\sum_{i}(\hat{Z}_{i}^{2}\hat{U}_{i}^{2}-\tilde{z}_{i}^{2}\tilde{u}_{i}^{2})}
    +\abs*{n\sum_{i}({Z}_{i}^{2}{U}_{i}^{2}-\tilde{z}_{i}^{2}\tilde{u}_{i}^{2})}
    +o_{p}(1).
  \end{equation}
  The second term is of the order $o_{p}(1)$ by condition (iii), so it suffices
  to show that the first term is of the order $o_{p}(1)$.

  Let $\check{U}_{i}=\hat{U}_{i}-\tilde{u}_{i}$ and
  $\check{D}_{i}=\hat{D}_{i}-\tilde{d}_{i}$, so we may write
  $n\hat{Z}_{i}\hat{U}_{i}=\frac{\ddot{D}_{i}\check{U}_{i}}{\ddot{D}'\ddot{D}/n}
  +\frac{\ddot{D}_{i}\tilde{u}_{i}}{\ddot{D}'\ddot{D}/n}
  -\frac{\check{D}_{i}\check{U}_{i}+\check{D}_{i}\tilde{u}_{i}+\tilde{d}_{i}\check{U}_{i}}{\hat{D}'\hat{D}/n}
  -\frac{\tilde{d}_{i}\tilde{u}_{i}}{\hat{D}'\hat{D}/n}$ and
  $n\tilde{z}_{i}\tilde{u}_{i}=\frac{\tilde{u}_{i}\ddot{D}_{i}}{\ddot{D}'\ddot{D}/n}
  -\frac{\tilde{u}_{i}\tilde{d}_{i}}{\tilde{d}'\tilde{d}/n}$. Plugging these
  expressions into the first term in \cref{eq:var_cons2} and expanding the term
  yields
  \begin{multline*}
    n\sum_{i}(\hat{Z}_{i}^{2}\hat{U}_{i}^{2}-
    \tilde{z}_{i}^{2}\tilde{u}_{i}^{2})= \frac{\frac{1}{n}\sum_{i}
      [(\check{D}_{i}\check{U}_{i}+\check{D}_{i}\tilde{u}_{i}+\tilde{d}_{i}\check{U}_{i})^{2}
      +2(\check{D}_{i}\check{U}_{i}+\check{D}_{i}\tilde{u}_{i}+\tilde{d}_{i}\check{U}_{i})
      \tilde{d}_{i}\tilde{u}_{i}]}{(\hat{D}'\hat{D}/n)^{2}}\\
    + \frac{\frac{1}{n}\sum_{i}(\ddot{D}_{i}^{2}\check{U}_{i}^{2}
      +2\ddot{D}_{i}^{2}\check{U}_{i}\tilde{u}_{i})}{(\ddot{D}'\ddot{D}/n)^{2}}
    -\frac{2}{n}\sum_{i}\frac{(\ddot{D}_{i}\check{U}_{i}+\ddot{D}_{i}\tilde{u}_{i})
      (\check{D}_{i}\check{U}_{i}+\check{D}_{i}\tilde{u}_{i}+\tilde{d}_{i}\check{U}_{i})
      +\ddot{D}_{i}\check{U}_{i}\tilde{d}_{i}\tilde{u}_{i}
    }{\ddot{D}'\ddot{D}/n\cdot\hat{D}'\hat{D}/n}
    \\
    +\frac{1}{n}\sum_{i}\tilde{d}_{i}^{2}\tilde{u}_{i}^{2}\left[\frac{1}{(\hat{D}'\hat{D}/n)^{2}}
      -\frac{1}{(\tilde{d}'\tilde{d}/n)^{2}}\right]
    +2\frac{\frac{1}{n}\sum_{i}\tilde{d}_{i}\ddot{D}_{i}\tilde{u}_{i}^{2}}{\ddot{D}'\ddot{D}/n}
    \left[\frac{1}{\tilde{d}'\tilde{d}/n}-\frac{1}{\hat{D}'\hat{D}/n}\right].
  \end{multline*}
  Using the inequality $\abs{ab}\leq (a^{2}+b^{2})/2$, \cref{eq:mean_ddotd}, and
  the fact that by condition (iv),
  $\hat{D}'\hat{D}/n=\tilde{d}'\tilde{d}/n+o_{p}(1)$ and that
  $\tilde{d}'\tilde{d}/n\asymp_{p}1$ by condition (i) and
  \Cref{assumption:an_long} (ii) and (iii), it follows that the right-hand side
  is smaller than
  \begin{multline*}
    \abs*{n\sum_{i}(\hat{Z}_{i}^{2}\hat{U}_{i}^{2}-
      \tilde{z}_{i}^{2}\tilde{u}_{i}^{2})}\leq
    O_{p}(1)\frac{1}{n}\sum_{i}[\check{D}_{i}^{2}\check{U}_{i}^{2}
    +\check{D}_{i}^{2}\tilde{u}_{i}^{2}+\tilde{d}_{i}^{2}\check{U}_{i}^{2}
    + \abs{\check{D}_{i}\tilde{d}_{i}}\tilde{u}_{i}^{2}+\tilde{d}_{i}^{2}\abs{\check{U}_{i}\tilde{u}_{i}}
    ] \\
    + O_{p}(1){\frac{1}{n}\sum_{i}
      [\ddot{D}_{i}^{2}\check{U}_{i}^{2}+(\ddot{D}_{i}^{4}+\tilde{u}_{i}^{2})\abs{\check{U}_{i}}]}\\
    +O_{p}(1)\frac{1}{n}\sum_{i}[
    \check{U}_{i}^{2}\check{D}_{i}^{2}+\ddot{D}_{i}^{2}\check{U}_{i}^{2}+\check{D}_{i}^{2}\tilde{u}_{i}^{2}
    +\tilde{d}_{i}^{2}\check{U}_{i}^{2}
    +\tilde{u}_{i}^{2}\abs{\ddot{D}_{i}\check{D}_{i}}
    +\abs{\check{U}_{i}}(\ddot{D}_{i}^{2}
    +\tilde{d}_{i}^{2}\tilde{u}_{i}^{2})]
    \\
    +o_{p}(1)\frac{1}{n}\sum_{i}\tilde{d}_{i}^{2}\tilde{u}_{i}^{2}
    +o_{p}(1)\frac{1}{n}\sum_{i}(\ddot{D}_{i}^{2}+\tilde{d}_{i}^{2})\tilde{u}_{i}^{2}\\
    \leq
    o_{p}(1)\frac{1}{n}\sum_{i}[1+\tilde{u}_{i}^{2}+\tilde{d}_{i}^{2}
    + \tilde{d}_{i}^{2}\tilde{u}_{i}^{2}
    +\ddot{D}_{i}^{4} +\ddot{D}_{i}^{2}\tilde{u}_{i}^{2}],
  \end{multline*}
  where the second inequality uses condition (iv) and the inequality
  $2\abs{a}\leq 1+a^{2}$. We now show that each summand is bounded in
  probability. By condition (i), and \Cref{assumption:an_long}~(ii),
  $\frac{1}{n}\sum_{i}(\tilde{u}_{i}^{2}+\tilde{d}_{i}^{2})=
  \frac{1}{n}\sum_{i}(U_{i}^{2}+D_{i}^{2})+o_{p}(1)=O_{p}(1)$. Next,
  \begin{equation*}
    \frac{1}{n}\sum_{i}\tilde{d}_{i}^{2}\tilde{u}_{i}^{2}
    \leq
    \frac{4}{n}\sum_{i}(\tilde{D}_{i}^{2}+r_{\delta, i}^{2})(r_{\gamma, i}^{2}+U_{i}^{2})
  \end{equation*}
  which by \Cref{assumption:an_long}~(ii) and conditions (i) and (ii) has
  expectation bounded by a constant times
  $\frac{1}{n}\sum_{i=1}^{n}E[1+r_{\gamma, i}^{2}+r_{\delta, i}^{2}+r_{\delta,
    i}^{2}r_{\gamma, i}^{2}]=O(1)$, so
  $\frac{1}{n}\sum_{i}\tilde{d}_{i}^{2}\tilde{u}_{i}^{2}=O_{p}(1)$ by Markov's
  inequality. Likewise, it follows from~\cref{eq:ddot4} that
  $\frac{1}{n}\sum_{i}\ddot{D}_{i}^{4}=O_{p}(1)$. Finally, by
  \Cref{assumption:an_long}~(ii), and the bound
  $E[\ddot{D}_{i}^{2}\mid W]\leq K M_{ii}\leq K$, we have
  $\frac{1}{n}\sum_{i}E[\ddot{D}_{i}^{2}\tilde{u}_{i}^{2}]\leq
  2\frac{1}{n}\sum_{i}E[\ddot{D}_{i}^{2}(K+r_{\gamma}^{2})]\leq
  2E[K(K+r_{\gamma}^{2})]=O(1)$, so that by Markov's inequality, the last term
  is also bounded in probability.
\end{proof}

\begin{proof}[Proof of~\Cref{theorem:non-normal-errors_epe}]
  Letting $r=(I-P_{\mathcal{S}^{*}})X\alpha$, we may write
  \begin{equation}\label{eq:fs_decomp}
    \mathcal{F} =\epsilon'P\epsilon- \epsilon'P_{\mathcal{S}^{*}}\epsilon
    +2\epsilon'r+r'r
    =\epsilon'P\epsilon+O_{p}(s+\norm{r}_{2}+\norm{r}_{2}^{2}),
  \end{equation}
  where the second equality follows since the second term satisfies
  \begin{equation}\label{eq:eps_Ps_eps}
    \epsilon'P_{\mathcal{S}^{*}}\epsilon=E[\epsilon'P_{\mathcal{S}^{*}}\epsilon]
    +O_{p}(s^{1/2})\leq Ks+O_{p}(s^{1/2})
  \end{equation}
  by \Cref{theorem:quadratic_form} and condition (a), and the
  third term is mean zero with variance bounded by $K\norm{r}^{2}$ by condition
  (a). Furthermore, by conditions (b) and (d),
  \begin{equation*}
    \omega^{2}:=2\sum_{i\neq j}E[\epsilon_{i}^{2}\epsilon_{j}^{2}\mid X]P_{ij}^{2}\geq
    2\sum_{i\neq j}P_{ij}^{2}/K^{2}\geq
    \frac{2}{K^{2}}\trace(P)(1-\max_{i}P_{ii})\geq 2\trace(P)/K^{3}.
  \end{equation*}
  Hence, by condition (c), condition (a) of \Cref{sec:qf_clt} holds. Since
  condition (b) of \Cref{sec:qf_clt} holds by condition (a), we can apply
  \Cref{sec:qf_clt} to $\epsilon'P\epsilon$ to yield
  $(\epsilon'P\epsilon-\sum_{i=1}^{n}\epsilon_{i}^{2}P_{ii})/\omega
  \overset{d}{\to}\mathcal{N}(0,1)$. Combining this result with
  \cref{eq:fs_decomp} and \Cref{assumption:approx_sparsity} yields
  \begin{equation}\label{eq:ePe_normal}
    \frac{\mathcal{F}-\sum_{i=1}^{n}\epsilon_{i}^{2}P_{ii}}{\omega}
    =
    \frac{\epsilon'P\epsilon-\sum_{i=1}^{n}\epsilon_{i}^{2}P_{ii}}{\omega}
    +o_{p}(1)
    \overset{d}{\to}\mathcal{N}(0,1).
  \end{equation}
  Furthermore, letting $\psi_{i}=\epsilon_{i}^{2}-E[\epsilon_{i}^{2}\mid X]$,
  and using the identity
  $\epsilon_{i}^{2}\epsilon_{j}^{2}-E[\epsilon^{2}_{i}\epsilon^{2}_{j}\mid X] =
  \psi_{i}\psi_{j}+\psi_{i}E[\epsilon_{j}^{2}\mid X]+E[\epsilon_{i}^{2}\mid
  X]\psi_{j}$, we have
  \begin{multline*}
    \abs*{\frac{2\sum_{i\neq
          j}\epsilon_{i}^{2}\epsilon_{j}^{2}P_{ij}^{2}-\omega^{2}}{\omega^{2}}}
    \leq \frac{K^{3}}{\trace(P)}
    \abs*{\sum_{i \neq j}
      (\epsilon_{i}^{2}\epsilon_{j}^{2}-E[\epsilon^{2}_{i}\epsilon^{2}_{j}\mid X])P_{ij}^{2}}\\
    \leq
     \frac{2K^{3}}{\trace(P)}\abs*{\sum_{i< j}\psi_{i}\psi_{j}P_{ij}^{2}}+
     \frac{2K^{3}}{\trace(P)}\abs*{\sum_{i \neq j}\psi_{i}E[\epsilon_{j}^{2}\mid X]P_{ij}^{2}}.
  \end{multline*}
  By condition (a), conditional on $X$, the term inside the first absolute value
  function is mean zero with variance bounded by a constant times
  $\sum_{i, j}P_{ij}^{4}\leq \trace(P)$, and the term inside the second absolute
  value function is also mean zero with variance bounded by a constant times
  $\sum_{i=1}^{n}(\sum_{j=1}^{n}P_{ij}^{2})^{2}\leq \trace(P)$, so that the
  above display is $O_{p}(\trace(P)^{-1/2})=o_{p}(1)$ by Markov's inequality.
  Combining this result with \cref{eq:ePe_normal} then yields the first claim.

  To prove the second claim in \Cref{theorem:non-normal-errors_epe}, it suffices
  to show that
  \begin{equation*}
    \norm{\yy-X\hat{\alpha}}_{2}^{2}-
    \yy'(I-P_{\mathcal{S}^{*}})\yy=O_{p}(s\log(p)),
  \end{equation*}
  that
  \begin{equation}\label{eq:replace_bias_correction}
    \sum_{i}P_{ii}(\hat{\epsilon}_{i}^{2}-\epsilon_{i}^{2})=o_{p}(p^{1/2})
  \end{equation}
  and that
  \begin{equation}\label{eq:replace_variance}
    \sum_{i\neq j}
      (\hat{\epsilon}_{i}^{2}\hat{\epsilon}_{j}^{2}-\epsilon_{i}^{2}\epsilon_{j}^{2})
      P_{ij}^{2}=o_{p}(p).
  \end{equation}
  The first assertion follows by \Cref{theorem:F_approx_clust}. To show
  \cref{eq:replace_bias_correction}, decompose
  $\hat{\epsilon}_{i} =\epsilon_{i}-\tilde{f}_{i}
  =\epsilon_{i}+r_{i}-X_{i}'\tilde{\alpha}$, where
  $\tilde{f}_{i}=X_{i}'(\hat{\alpha}-\alpha)$, and
  $\tilde{\alpha}=\hat{\alpha}-(X_{\mathcal{S}^{*}}'X_{\mathcal{S}^{*}})^{-1}X_{\mathcal{S}^{*}}'X\alpha$,
  so that we may write
  \begin{multline*}
    \abs*{\sum_{i}P_{ii}(\hat{\epsilon}_{i}^{2}-\epsilon_{i}^{2})}
    =\abs*{\sum_{i}P_{ii}\tilde{f}_{i}^{2}
      +2\sum_{i}P_{ii}\epsilon_{i}r_{i}
      -2\sum_{i}P_{ii}\epsilon_{i}X_{i}'\tilde{\alpha}}\\
    \preceq_{p}
    s \log(p)
    +\norm{r}_{2}
    +2\norm*{\sum_{i}P_{ii}\epsilon_{i}X_{i}}_{\infty}\norm{\tilde{\alpha}}_{1}\preceq_{p}
    s \log(p)
    +\norm{r}_{2},
  \end{multline*}
  where the first inequality follows by condition (i) and Markov's inequality,
  and since conditional on $X$, the second term is mean zero with standard
  deviation bounded by a constant times $\norm{r}_{2}$. The second inequality
  applies conditions (ii) and (iv). \Cref{eq:replace_bias_correction} then
  follows by \Cref{assumption:approx_sparsity}.

  To show \cref{eq:replace_variance}, write the left-hand side as
  \begin{multline*}
    \sum_{i\neq j}
    (\hat{\epsilon}_{i}^{2}\hat{\epsilon_{j}}^{2}-\epsilon_{i}^{2}\epsilon_{j}^{2})
    P_{ij}^{2}\\
    =
    \sum_{i\neq j}
    \tilde{f}_{i}^{2}\tilde{f}_{j}^{2}P_{ij}^{2}
    +2 \sum_{i\neq j}
    \epsilon_{i}^{2}\tilde{f}_{j}^{2}P_{ij}^{2}
    -4 \sum_{i\neq j}
    \epsilon_{j}\tilde{f}_{j}\tilde{f}_{i}^{2}P_{ij}^{2}
    -4 \sum_{i\neq j}
    \epsilon_{i}\tilde{f}_{i}\epsilon_{j}^{2}P_{ij}^{2}
    +4 \sum_{i\neq j}
    \epsilon_{i}\tilde{f}_{i}\epsilon_{j}\tilde{f}_{j}P_{ij}^{2}
    \\
    =:T_{1}+2T_{2}-4T_{3}-4T_{4}+4T_{5}= O_{p}(s^{2}\log(p)^{2}
    +\norm{r}_{2}^{2})+o_{p}(p),
  \end{multline*}
  where the bounds follow by bounding each term, as derived next. First, by
  condition (i), $T_{1}$ is bounded by a constant times
  $\norm{X(\hat{\alpha}-\alpha)}_{2}^{4}\preceq_{p} s^{2}\log(p)^{2}$. Second,
  \begin{equation*}
    T_{2}\leq 2\sum_{i\neq j}\epsilon_{i}^{2}r_{j}^{2} P_{ij}^{2}
    +2\sum_{i\neq j}\epsilon_{i}^{2}(X_{j}'\tilde{\alpha})^{2}P_{ij}^{2}
    = O_{p}(\norm{r}_{2}^{2})+
    o_{p}(p).
  \end{equation*}
  where the second equality follows by Markov's inequality, since the first
  term, conditional on $X$, has expectation bounded by a constant times
  $\sum_{j}r_{j}^{2} P_{jj}\leq \norm{r}_{2}^{2}$, and the second term is
  bounded by
  $\max_{j}(X_{j}'\tilde{\alpha})^{2}\cdot \sum_{i\neq
    j}\epsilon_{i}^{2}P_{ij}^{2}=o_{p}(1)\cdot O_{p}(p)$, since
  $\sum_{i\neq j}E[\epsilon_{i}^{2}P_{ij}^{2}]$ is bounded by a constant times
  $\sum_{i\neq j}P_{ij}^{2}\leq p$, and by conditions (ii) and (iii),
  \begin{equation}\label{eq:maxXalpha}
    \max_{j}\abs{X_{j}'\tilde{\alpha}}\leq \max_{j}\norm{X_{j}}_{\infty}\norm{\tilde{\alpha}}_{1}
    =o_{p}(1).
  \end{equation}
  Third,
  \begin{equation*}
    \abs{T_{3}}
    \leq
    2\sum_{i\neq j}\epsilon_{j}^{2}\tilde{f}_{i}^{2}P_{ij}^{2}+
    2\sum_{i\neq j}\tilde{f}_{j}^{2}\tilde{f}_{i}^{2}P_{ij}^{2}=2T_{1}+2T_{2}.
  \end{equation*}
  Fourth,
  \begin{multline*}
    \abs{T_{4}}
    \leq \abs*{\sum_{i\neq j}\epsilon_{i}X_{i}'\tilde{\alpha}\epsilon_{j}^{2}P_{ij}^{2}}
    +2\abs*{\sum_{i\neq j}\epsilon_{i}r_{i}(\epsilon_{j}^{2}-E[\epsilon_{j}^{2}\mid X])P_{ij}^{2}}
    +\abs*{\sum_{i\neq j}\epsilon_{i}r_{i} E[\epsilon_{j}^{2}\mid X]P_{ij}^{2}}\\
    =
    o_{p}(p)+O_{p}(\norm{r}_{2}).
  \end{multline*}
  Here the conclusion follows since the first term is bounded by
  $\max_{i}\abs{X_{i}'\tilde{\alpha}} \sum_{i\neq
    j}\abs{\epsilon_{i}}\epsilon_{j}^{2}P_{ij}^{2}=o_{p}(p)$ by
  \cref{eq:maxXalpha} and the fact that
  $E\sum_{i\neq j}\abs{\epsilon_{i}}\epsilon_{j}^{2}P_{ij}^{2}\preceq
  \sum_{i, j}P_{ij}^{2}\leq p$, the second term is mean zero with variance bounded by
  a constant times $\sum_{i, j}r_{i}^{2}P_{ij}^{4}\leq \norm{r}_{2}^{2}$, and the
  third term is mean zero with variance bounded by a constant times
  $\sum_{i}r_{i}^{2}P_{ii}^{2}\leq \norm{r}_{2}^{2}$. Finally,
  \begin{equation*}
    T_{5}=
    2\sum_{i< j}\epsilon_{i}\epsilon_{j}r_{j}r_{i}P_{ij}^{2}
    +
    2\sum_{i\neq j}\epsilon_{i}r_{i}\epsilon_{j}X_{j}'\tilde{\alpha}
    P_{ij}^{2}
    + \sum_{i\neq j}\epsilon_{i}\epsilon_{j}(X_{j}'\tilde{\alpha})
    (X_{i}'\tilde{\alpha})P_{ij}^{2}.
  \end{equation*}
  The first term is mean zero with variance bounded by a constant times
  $\sum_{i< j}r_{j}^{2}r_{i}^{2}P_{ij}^{4}\leq
  \norm{r}_{2}^{4}$. The second term is bounded by a constant times
  $\max_{j}\abs{X_{j}'\tilde{\alpha}} \cdot\sum_{i\neq
    j}r_{i}^{2}\abs{\epsilon_{j}} P_{ij}^{2} +
  \max_{j}\abs{X_{j}'\tilde{\alpha}}\cdot \sum_{i\neq j}\epsilon_{i}^{2}\abs{\epsilon_{j}}
  P_{ij}^{2}=o_{p}(\norm{r}_{2}^{2}+p)$, and the last term is
  bounded by
  $\max_{j}\abs{X_{j}'\tilde{\alpha}}^{2}\cdot \sum_{i\neq
    j}\abs{\epsilon_{i}\epsilon_{j}}P_{ij}^{2}=o_{p}(p)$. Hence,
  $T_{5}=O_{p}(\norm{r}^{2})+o_{p}(p)$, concluding the proof.
\end{proof}

\section{Finite-sample size properties of sparsity tests}\label{sec:monte-carlo}

We now conduct a simple Monte Carlo simulation to assess the finite-sample
performance of the Hausman and residual tests considered in
\Cref{sec:testing-sparsity}.

The control matrix $W$ in these simulations corresponds to one of seven designs,
each corresponding to one of the seven specifications reported in the first two
columns of \Cref{tab:t1}. The treatment and outcome vectors $D$ and $Y$ are
drawn from independent normal distributions with zero means and homoskedastic
variances $\sigma^{2}_{Y}$ and $\sigma^{2}_{D}$. For each control matrix, we
consider two values for these variances: either the variance of the outcome and
treatment, or else the variances of OLS residuals in each empirical
specification. The first method yields standard deviation values
$(\sigma_{Y}, \sigma_{D})$ given by
$\{(0.082, 0.083), (0.052, 0.095), (0.264, 0.083)\}$ in the three BCH
specifications, $(14.2, 2.1)$ in the \textcite{ferrara22} application, and
$\{(37.5, 1.1), \allowbreak (49.7 1.0), (50.0 0.8)\}$ in the three
\textcite{enke20} specifications. The second method gives standard deviation
values $\{(0.045, 0.004), (0.030, 0.009), (0.197, 0.004)\}$ in the three BCH
specifications, $(3.1, 0.7)$ in the \textcite{ferrara22} application, and
$\{(27.2, 0.5), (29.8, 0.6), (31.4, 0.5)\}$ in the three \textcite{enke20}
specifications.

In other words, each design corresponds exactly to the empirical specifications
we studied in \Cref{sec:empir-tests-spars}, except we replace the outcome and
treatment with mean zero Gaussian variables. The parameter values
$\beta, \gamma$ and $\delta$ in the outcome and propensity score regressions in
\cref{eq:outcome,eq:p_score} therefore all equal $0$, so that the sparsity
assumption holds trivially. For each design, we apply the same three tests as in
\Cref{tab:t3}, with the same implementation. In particular, the tests allow for
heteroskedasticity in the \textcite{enke20} specification, and for clustering in
the other two specifications. To benchmark the performance of the tests, we also
compare the test to an oracle which replaces the \ac{SBE} estimates with a
constrained \ac{OLS} estimate that sets all penalized coefficients to
zero.\footnote{The unpenalized controls comprise 12 year fixed effects in BCH,
  county fixed effects \textcite{ferrara22}, and the intercept in
  \textcite{enke20}.}

\begin{table}
  \renewcommand*{\arraystretch}{1.2}
  \begin{threeparttable}
    \caption{Empirical null rejection rates (in percentages) of sparsity tests
      with nominal level 5\% in simulations.}\label{tab:simulation}
  \begin{tabular*}{0.95\linewidth}{@{\extracolsep{\fill}}ll@{} S@{}S@{}S@{}S@{}}
    & & \multicolumn{2}{c}{Large variance} & \multicolumn{2}{c}{Small variance}\\
    \cmidrule(rl){3-4}\cmidrule(rl){5-6}
    && \multicolumn{1}{c}{\ac{SBE}} & \multicolumn{1}{c}{Oracle}
                                           & \multicolumn{1}{c}{\ac{SBE}} & \multicolumn{1}{c}{Oracle}\\
    Outcome & Test & \multicolumn{1}{c}{(1)} & \multicolumn{1}{c}{(2)}
            & \multicolumn{1}{c}{(3)} & \multicolumn{1}{c}{(4)} \\
    \midrule
      \csname @@input\endcsname ./table4.tex
    \bottomrule
  \end{tabular*}
  \begin{tablenotes}
  \item{}\footnotesize\emph{Notes}: H\@: Hausman, OR\@: residual test based on
    outcome regression, PR\@: residual test based on propensity score
    regression. Oracle corresponds to a constrained \ac{OLS} estimator that sets
    all penalized coefficients to zero. The control matrix $W$ corresponds to
    the control matrix in Cols.~1--2 of \Cref{tab:t1}. In Cols.~1--2, the
    outcome and treatment variance correspond to the variance of the outcome and
    treatment in the corresponding specification in Cols.~1--2 of \Cref{tab:t1}.
    In Cols.~3--4, the outcome and treatment variances correspond to variances of
    OLS residuals in these specifications. 10,000 simulation draws.
  \end{tablenotes}
\end{threeparttable}
\end{table}

We report the empirical null rejection rates of the sparsity assumption for each
test in \Cref{tab:simulation}. The rejection rates do not exceed the nominal 5\%
level by more than one percentage point, and do not deviate from the oracle
rejection rates by more than two percentage points in any specification. Both
the feasible and the oracle version of the residual test is conservative in
several specifications: this appears to be a consequence of the fact that the
mean of the infeasible test statistic $\mathcal{F}$ is given by
$E[\sum_{i}\epsilon_{i}^{2}P_{ii}-\sum_{i}\epsilon_{i}^{2}P_{\mathcal{S}^{*},
  ii}]$, but the recentering in \cref{eq:infeasible} only accounts for the first
term, since the second term is asymptotically negligible. This may render the
test conservative in finite samples when the true sparsity level is not zero, or
when unpenalized covariates (that enter estimate of the sparsity index
$\mathcal{S}^{*}$) are present.

\section{Literature review details}\label{sec:deta-liter-revi}

We ordered all empirical papers that cite one of the 8 \ac{SBE} methodological
papers cited in the introduction
\parencite{bch14,JaMo14,vbrd14,ZhZh14,aiw18,bcfh17,ccddhnr18,farrell15} by the
number of citations on Google Scholar as of July 4, 2024. Our sample consists of
50 papers with the highest number of citations. Of these, 3 papers use lasso or
post-lasso rather than \acp{SBE}, and we drop these from the sample. The
remaining 47 papers all use variants of the \textcite{bch14}
post-double-selection procedure.\footnote{A few papers apply this approach in
  the context of non-linear regressions (e.g.\ logit regressions), or
  instrumental variables regressions. As our discussion about the specification
  of the control matrix $W$ is also applicable in these contexts, we keep such
  papers in the sample.}

For the \LRnRep\ papers with available replication code, we have replicated all
\LRspecRep\ regression specifications that employ \acp{SBE}. For these papers,
the number of controls $p$ reported in \Cref{tab:t_lit1} correspond to the
dimension of the control matrix after dropping columns with no variation and
repeated columns. Additionally, we could determine the ratio $p/n$ in
\LRnDiff\ of the papers without replication code either because they
explicitly mentioned the value of $p$, or else provided a clear enough
description of the control matrix construction. Columns (1) to (3) in
\Cref{tab:t_lit1} report the distribution of this ratio in this sample of
\LRnAll\ papers. Each specification is weighted so that each paper receives
an equal weight: if, say, a paper features 10 specifications that employ
\acp{SBE}, each receives weight $1/10$.

For coding whether the control matrix basis is specified in the description in
the text, as reported in col~(4) of \Cref{tab:t_lit1}, we assumed that, unless
otherwise explicitly stated, variables were not normalized prior to taking
interactions and powers or other non-linear transformations; if this was the
only normalization issue in the specification, we coded it as $1$. Similarly, we
assumed that specifications featuring interacting with a binary variable kept
one interaction and the main effect (even though keeping both interactions and
dropping the main effect also leads to the same column space); we coded such
specifications as $1$. If the paper did not explicitly list the full set of
control variables, or the set of variables contained categorical variables and
the paper did not make the base category explicit, this was coded as $0$.

\end{appendices}

\phantomsection\addcontentsline{toc}{section}{References}
\printbibliography

\end{document}